\documentclass{article}

\usepackage{PRIMEarxiv}
\usepackage{amsmath}
\usepackage[utf8]{inputenc} 
\usepackage[T1]{fontenc}    
\usepackage{hyperref}       
\usepackage{url}            
\usepackage{booktabs}       
\usepackage{amsfonts}       
\usepackage{nicefrac}       
\usepackage{microtype}      
\usepackage{lipsum}
\usepackage{fancyhdr}       
\usepackage{graphicx}       
\graphicspath{{media/}}     
\usepackage{float} 

\usepackage{graphicx}
\usepackage{multirow}
\usepackage{algorithm,algpseudocode}
\usepackage{algorithmicx}
\usepackage{subcaption}
\usepackage{float}
\usepackage{amsthm}
\newtheorem{prop}{Proposition}

\title{Unsupervised Ensemble Based Deep Learning Approach for Attack Detection in IoT Network}

\author{
 Mir Shahnawaz Ahmed, Shahid M Shah \\
Communication Control  Learning Lab, \\Department of Electronics \& Communication Engineering,\\ National Institute of Technology Srinagar\\
  \texttt{\{mirshahnawaz888@gmail.com, shahidshah@nitsri.net\}} \\
}

\begin{document}
\maketitle

\begin{abstract}
The Internet of Things (IoT) has altered living by controlling devices/things over the Internet. IoT has specified many smart solutions for daily problems, transforming cyber-physical systems (CPS) and other classical fields into smart regions. Most of the edge devices that make up the Internet of Things have very minimal processing power. To bring down the IoT network, attackers can utilise these devices to conduct a variety of network attacks. In addition, as more and more IoT devices are added, the potential for new and unknown threats grows exponentially. For this reason, an intelligent security framework for IoT networks must be developed that can identify such threats. In this paper, we have developed an unsupervised ensemble learning model that is able to detect new or unknown attacks in an IoT network from an unlabelled dataset. The system-generated labelled dataset is used to train a deep learning model to detect IoT network attacks. Additionally, the research presents a feature selection mechanism for identifying the most relevant aspects in the dataset for detecting attacks. The study shows that the suggested model is able to identify the unlabelled IoT network datasets and DBN (Deep Belief Network) outperform the other models with a detection accuracy of 97.5\% and a false alarm rate of 2.3\% when  trained using labelled dataset supplied by the proposed approach.
\end{abstract}


\section{Introduction}
\label{IntroSection}
The Internet of Things (IoT) expands the Internet's edge by integrating new terminal devices and resources at the network's edge. It is one of the fastest-growing and most extensively utilised technologies on the internet. IoT constitutes integration of millions of devices that can communicate with one another to solve a wide range of real-world challenges by providing innovative and cost-effective solutions. Smart metering, smart transportation systems, smart homes, smart medical care, smart agriculture etc. are growing as a result of IoT technology \cite{javed2018internet}. Apart from these applications, IoT devices are widely used for the purpose of surveillance \cite{ciuonzo2021distributed}.

The sensors used in IoT networks/CPS continuously monitor their surroundings, thereby generating huge volumes of data. These massive volumes of data are saved in data-centers, some of which may contain sensitive information about systems/individuals, thus making IoT networks/CPS vulnerable to various cyber attacks. More apparent and unseen attacks are developing like Denial of Service (DoS), Remote to Local (R2L), Brute Force, Probing (Probe), User to Root (U2R), Man-in-the-middle, Scanning, Ransomware, Password attack, etc., wreaking havoc and creating irreversible damage \cite{meneghello2019iot}. Also, the use of large numbers of heterogeneous IoT devices in a network increases the attack surface. Attackers generally target the data-centers because they have large volumes of diverse information. In addition, IoT network devices have limited storage and computational capabilities which prohibit them from detecting and fighting against possible online cyber threats. Also, due to the increase in the attack surfaces for IoT networks, the detection of new/unknown attacks increases the challenge of attack detection \cite{hassan2019current}. IoT networks might be severely harmed by even a moderate security attack, thereby hindering the applicability of IoT networks in different fields. IoT networks generate enormous volumes of network traffic and have more number of edge devices of varying types than the traditional computer networks, making IoT network security concerns more complicated and difficult. Thus, it is evident that a smart and intelligent attack detection mechanism has to be developed which is able to detect even new/unknown attacks in IoT networks. The characteristics of machine learning techniques and their recent advancements make them suitable for the purpose of attack detection in IoT networks. The list of acronyms used in the paper are described in table \ref{acronyumTable}.

The contributions of this paper are summarized as:
\begin{itemize}
\item We propose an amalgamation of correlation coefficient and LASSO regression for the purpose of feature selection to detect attacks in an IoT network. The amalgamation is done after a thorough analysis of the performance metrics and identification of critical parameters for correlation based and LASSO regression based feature selection techniques separately. The fusion of both the techniques is done in such a way that the overall performance of the system is increased by selecting most relevant features from the available network datasets.
\item To identify new/unknown attacks, the model should be able to identify hidden patterns in the network dataset and cluster the network traffic into malicious and non-malicious. To achieve this, we have used a group of clustering algorithms. The output of each clustering algorithm, for a particular IoT network traffic, is combined using a weighted voting technique to predict the class label (malicious/non-malicious) with increased accuracy. The weights associated with the output of each clustering technique have been calculated after a detailed performance analysis. The combination of clustering techniques with the weighted voting gives rise to an ensemble machine learning approach which converts an unlabelled dataset (representing IoT network traffic with malicious attacks) into labelled dataset, thereby proposing an unsupervised mechanism which is able to identify new/unknown attacks in the IoT network traffic.
\item The labelled dataset generated by the proposed model is utilized to train a deep learning model for attack detection in an IoT network. To identify which deep learning model is best suited for the purpose, a performance analysis of different deep learning models (LSTM, MLP and DBN) has been done, and an efficient one is proposed to be deployed for detecting attacks in an IoT network.
\end{itemize} 

\begin{table}
\caption{List of Acronyms} 
\centering
\begin{tabular}
{p{0.1\linewidth}p{0.3\linewidth}p{0.1\linewidth}p{0.3\linewidth}}\hline

Acronyms &
Description&
Acronyms &
Description\\\hline
IoT &
Internet of Things&

RPL & 
Routing Protocol for Low-Power and Lossy Networks\\

CPS &
Cyber-Physical Systems&

DoS &
Denial of Service attack\\

DDoS &
Distributed Denial of Service attack&

R2L &
Remote to Local attack\\

Probe &
Probing attack&

U2R &
User to Root attack\\

LSTM &
Long Short Term Memory&

MLP &
Multilayer Perceptron\\

DBN &
Deep Belief Networks&

PCA &
Principal Component Analysis\\

IDS &
Intrusion Detection System&

CNN &
Convolution Neural Network\\

OPTICS &
Ordering Points to Identify the Clustering Structure&

ARI &
Adjusted Rand Index\\

AMI &
Adjusted Mutual Information&

FAR &
False Alarm Rate\\

MCC &
Matthews Correlation Coefficient&

SGD &
Stochastic Gradient Descent\\

RBM &
Restricted Boltzmann Machine&

ROC &
Receiver Operating Characteristic\\

PR &
Precision-Recall&

DODAG &
Destination Oriented Directed Acyclic Graph\\\hline

\end{tabular}
\label{acronyumTable} 
\end{table}

The rest of the paper is organised as: Section \ref{RelatedWorkSection} describes the recent literature in the direction of attack detection in IoT networks. Section \ref{ProposedModelSection} gives the detailed view of the proposed unsupervised ensemble learning model. The implementation of the proposed model along with the deep learning models used are discussed in section \ref{ResultsSection}. Finally, section \ref{ConclusionSection} concludes the paper and presents future directions.

\section{Related Work}
\label{RelatedWorkSection}
With the advancement in the field of IoT, many researchers have proposed various techniques for detecting malicious behaviour in IoT and hence creating a secure environment for implementing IoT network.

Xiao et al. in \cite{xiao2018iot} discuss various security mechanisms which use machine learning algorithms to achieve the security requirements of an IoT network, but most of them require the accurate knowledge of network and attack state, which is difficult in IoT networks. Furthermore various machine learning based security solutions have high computational complexities and require huge training data.
Lianbing et al. \cite{deng2019mobile} explains various types of Artificial Intelligence based Intrusion Detection techniques to detect malicious behaviours in IoT network and, finally proposed cloud architecture based intrusion detection system, which uses a combination of Fuzzy C-Means algorithm and Principal Component Analysis (PCA) approach. The proposed system has used PCA for dimensionality reduction of used dataset, in order to decrease the attack detection time, and Fuzzy C-Means algorithm is used for clustering purposes.
Shailendra et al. \cite{rathore2018semi} proposed a distributed attack detection framework for fog-based IoT which uses Semi-supervised learning approach. They also have used the Fuzzy C-Means algorithm to enhance the efficiency of the clustering mechanism and used Extreme Learning Machine (a feed-forward neural network with a single hidden layer). The performance of the system has been tested on NSL KDD dataset with an attack detection accuracy of 86.53\% and attack detection time of 11ms, but the attack detection accuracy of the system needs to be enhanced.
Sahay et al. \cite{sahay2018traffic} proposes a Multilayer Perceptron (MLP) based Misappropriation attack detection mechanism for IoT network. The authors describe the Misappropriation attack as decreased rank attack, where an attacker decreases the rank of a malicious network path in a DODAG, thereby deviating the traffic to the network path containing attacker nodes.DODAG is used by RPL for IoT network organization. The proposed approach has been implemented in Cooja simulator. It only detects one type of network attack, and the attack detection accuracy needs to be increased. Also, Sahay et al. \cite{sahay2018efficient} also presented a similar approach for detecting version number attack in RPL.
Mirsky et al. \cite{mirsky2018kitsune} proposes an ANN based unsupervised online IDS, which uses Autoencoders to reconstruct the network traffic for anomaly detection. Even though the authors claim that the proposed model is lightweight, but the resources requirement makes them computational infeasible for IoT networks.
Insider threats are one of the most difficult cyber-security concerns today, and they aren't effectively addressed by most security solutions \cite{homoliak2019insight}. An insider threat is caused by a legitimate user to jeopardise a network's data, policies, or devices in an unwanted or disruptive manner. These attacks are more common in the IoT environment and hard to detect, but intelligent machine learning techniques can be used to detect such attacks effectively \cite{ahmad2021mitigating}.
A hybrid intrusion detection mechanism was proposed by Ansam et al. in \cite{khraisat2019novel} which combines the advantages of signature and anomaly based  intrusion detection systems (IDS). The proposed mechanism first creates a decision tree using C5 Decision Tree Classifier to identify various known attacks based on predefined attack signatures. The unknown attacks are then detected using anomaly based IDS, which is trained using One-class Support Vector Machine. The combination of these two techniques result in a hybrid system which is able to detect DDOS, DOS, Reconnaissance and Key-logging attacks in BoT-IoT intrusion detection dataset. The proposed approach relies on predefined attack signatures for the purpose of training a machine learning model, thus making it an ineffective approach for detecting unknown attacks.
Venkatraman et al. in \cite{venkatraman2020adaptive} have also proposed a hybrid of signature based and anomaly based intrusion detection system, which uses Crowd-sourcing to gather various attack signatures and then use Automata based event processor to automatically detect malicious behaviours. The proposed attack detection system has been tested for smaller IoT networks.
Bovenzi et al. \cite{bovenzi2020hierarchical} proposes a MultiModal Deep AutoEncoder for detecting attacks in Bot-IoT dataset. The model has considerable performance but the use of AutoEncoders increase the complexity of the proposed model.
A blockchain based security framework has also been proposed by Sahay et al. \cite{sahay2020novel}, in which the vulnerabilities of RPL are explored to propose a secure attack detection module for an IoT network using blockchain. The proposed system is computationally complex because of blockchain and hence the complexity needs to be reduced.
Another blockchain based secure data sharing mechanism was proposed by veeramakali et al. \cite{veeramakali2021intelligent}, but here again the overall computational complexity of the system has to be reduced.
A Specification based Intrusion Detection was proposed by Jagadeesh Babu et al. \cite{babu2020sh}, which relies on a supervised learning approach to discover various transaction specifications for normal and abnormal behaviour. The discovered specifications are then further optimized using the proposed heuristics approach, which results in a system capable of identifying normal and abnormal (malicious) behaviour depending on the specifications of the network transactions performed by various IoT devices.
Aaisha et al. have mentioned in \cite{makkar2020efficient} that various devices in IoT networks prefer to search/access data using search engines available on the internet, but this way of accessing the data faces Web Spam as its biggest hindrance. To overcome this challenge they have proposed a deep learning based mechanism, which uses LSTM Network, that when incorporated in the search engine detects spam while calculating the page rank score for a particular web-page.
Eskandari et al. in \cite{eskandari2020passban} propose an Intrusion Detection System, which identifies new attacks (anomalies) by using one-class classification i.e., it analysis the characteristics of normal IoT network traffic and then detects anomalies if the characteristics of network traffic deviate. It uses Isolation Forest for isolating the malicious from non-malicious data, and Local Outlier Factor method for detecting the attacks. However, the false positive rate of the IDS is high. Also, the efficiency of the system has been calculated on a real testbed and the system needs new training sessions each time the underlying network changes.
Verma et al. \cite{verma2020machine} discusses the performance analysis of various supervised machine learning classifiers (Random Forest, Adaboost, Gradient Boosting, Extreme Gradient Boosting and Extremely Randomized Trees) for detecting DoS attack, among which Extreme Gradient Boosting outperforms the other, and also uses Raspberry Pi to find the response time for each classifier. The performance analysis is done using CIDDS-001, UNSW-NB15, and NSL-KDD datasets, by using all the available features. It is not able to detect new/unknown attacks.\\
Sahu et al. in \cite{sahu2021internet} presents a combination of Convolution Neural Network (CNN) and LSTM Model for attack detection in IoT. The proposed model uses CNN for the purposes of feature selection and trains a LSTM model with the selected features to detect attacks in the IoT network. Although the model achieves 92\% accuracy, it is computationally complex and requires more energy resources. Also the model is tested on real testbed with a smaller size of IoT network.
Also, Fotohi et al. in \cite{fotohi2021lightweight} proposes a lightweight clustering mechanism for detecting flooding attacks in IoT environment using ant colony optimization algorithm. The model efficacy has been tested using the NS-3 network simulator.

\begin{table}[H]
\centering
\caption{Comparison of some of the related work with the proposed model} 
\footnotesize
\begin{center}
\begin{tabular}
{p{0.07\linewidth}p{0.27\linewidth}p{0.09\linewidth}p{0.09\linewidth}p{0.08\linewidth}p{0.25\linewidth}}\hline

References &
Description &
Dataset Used &
No. of features used &
Learning Approach used &
Limitations\\\hline

Ambusaidi et al. \cite{ambusaidi2016building} &
The paper proposes a mutual information based feature selection mechanism and uses Least Square SVM to detect network attacks. &
NSL-KDD &
18 &
Supervised &
The mechanism only detects known attacks, and is unable to detect new/unknown network attacks. \\

Shone et al. \cite{shone2018deep} &
It uses stacked encoder phase of a deep auto-encoder, together with a random forest classifier to create a network intrusion detection system. &
NSL-KDD &
All &
Supervised &
It uses all the features available in the dataset and has an attack detection rate of 85.42\% and false alarm rate of 14.58\%, and is not able to detect new/unknown attacks.\\

Meidan et al. \cite{meidan2018n} &
The paper proposes a deep auto-encoder based detection system for IoT networks. It uses normal network traffic to train the model, and an anomaly is detected if the auto-encoder fails to reconstruct the data. &
NA (model tested on real IoT network) &
NA &
Supervised &
The proposed model was trained to detect botnet attacks only, and it is unable to detect new/unknown attacks.\\

Shailendra et al. \cite{rathore2018semi} &
The paper proposes a combination of Fuzzy C-means algorithm with Extreme Learning Machine to detect attacks. &
NSL-KDD &
All &
Semi-supervised &
Even-though the model is able to detect new/unknown attacks, but the model has attack detection accuracy of 86.53\% only. \\

Ansam et al. in \cite{khraisat2019novel} &
It proposes a hybrid IDS, which uses C5 decision tree to detect known attacks and One-class SVM to detect anomalies. &
Bot-IoT &
All &
Semi-supervised &
The proposed approach relies on predefined attack signatures for the purpose of training a machine learning model, thus making it an ineffective approach for detecting unknown attacks.\\

Jan et al. \cite{jan2019toward} &
It evaluates  SVM with different kernel functions to detect network attacks using variation in the traffic intensity. &
Dataset generated from CICIDS2017 using Poisson dist. &
NA &
Supervised &
The proposed model is only tested to detect those attacks which vary in network traffic intensities.\\

Zhou et al. \cite{zhou2020building} &
The paper proposes an ensemble approach that uses voting mechanism to combine the results of C4.5, Random Forest, and Forest by Penalizing Attributes techniques, to create an ensemble mechanism for detecting attacks. It also uses Bat algorithms with correlation based feature selection technique to retrieve relevant features from the dataset. &
NSL-KDD, AWID, and  CICIDS2017 &
10, 8, 13 (resp.)&
Supervised &
Unable to detect new/unknown attacks, and is computationally complex.\\

Srinivas and Patnaik \cite{srinivas2022clustering}&
They have used a metaheuristic quantum worm swarm optimization-based clustering technique to propose a secure routing protocol for mobile ad hoc networks, by selecting an optimal cluster head and secure routing path.&
Mobile ad hoc network created in MATLAB&
All&
Unsupervised&
It only detects routing attacks in mobile ad hoc networks.\\ 

Sahay et al. \cite{sahay2022holistic}&
The paper proposes a holistic approach that uses smart contract-fortified blockchain technology together with LSTM and CNN to detect routing attacks in an IoT network&
IoT traffic from Cooja simulator&
All&
Unsupervised&
It only detects known routing attacks in an IoT network using network specifications.\\\cline{2-6}
\vspace{0.1em}
\textbf{Proposed Work} &
\multicolumn{5}{p{0.9\linewidth}}{\vspace{0.1em}The paper proposes an unsupervised ensemble approach for clustering network traffic into malicious and non-malicious (thereby identifying new/unknown attacks) and uses system identified network traffic to train deep learning model to detect attacks in an IoT network. The proposed mechanism was tested on NSL-KDD and TON-IoT network datasets and has overall attack detection accuracy of 97.6\% and false alarm rate of 2.3\%. The paper also proposes feature selection mechanism (which selects 11 and 9 relevant features from respective datasets).}\\\hline

\end{tabular}
\end{center} 
\label{relatedWorkTable} 
\end{table}

Most of the above cited papers mainly use supervised machine learning approaches, which are efficient in detecting known attacks in an IoT network, but they rely on labelled datasets and are ineffective in detecting new/unknown attacks. Also, due to the increasing use of the IoT network, various new and heterogeneous devices get added in the network, which can be used by an attacker to launch new attacks. Thus, it is necessary to propose an intelligent attack detection system for IoT networks which uses an unsupervised machine learning approach to detect new/unknown attacks in an IoT network. Also, the performance of such systems can be improved by combining ensemble machine learning approaches with deep learning models. The comparison of some of the related work with the proposed model is shown in table \ref{relatedWorkTable}.

\section{Proposed Method}
\label{ProposedModelSection}
In this section we introduce the proposed model for the detection of attacks in an IoT network. The proposed model is shown in figure \ref{systemmodel}, which takes an unlabelled IoT network dataset (containing IoT network traffic: malicious as well as non-malicious) as input. The input unlabelled dataset is first pre-processed for any missing/redundant data. After data pre-processing, the value of each feature is being scaled using a standard scaler. The pre-processed and scaled dataset is then given to the proposed feature selection module to select relevant features for detecting network attacks. After feature selection, each data entry of the dataset (for selected features) is given to three clustering mechanisms (Mini batch K-means, Fuzzy C-means and OPTICS) simultaneously which cluster the data into malicious (represented by 1) or non-malicious (represented by 0). The output of each clustering algorithm is combined using a weighted voting, thereby predicting the most appropriate/accurate class label for the data entry. The combination of clustering algorithms and voting mechanism formulates the proposed ensemble learning model. So, after processing the entire unlabelled dataset by the ensemble learning model, a labelled dataset is generated. This process of converting the unlabelled IoT network dataset into a labelled dataset suffices that any new/unknown IoT network attacks can be also detected by the system. Due to the limited availability of computational resources in IoT devices, the proposed ensemble learning model can be deployed at a cloud layer and the labelled dataset generated by the system can be used to train a deep learning model for detecting new/unknown attacks in an IoT network. Hence at step-7 in figure \ref{systemmodel} we propose to use a deep learning model, which is to be trained using the system generated labelled network dataset. Finally, the trained deep learning model can be deployed in any IoT device at fog or edge layer of fog computing architecture. The deep learning model can be then updated timely at cloud layer using the proposed ensemble learning model for detecting new/unknown attacks that may be encountered in future.

\begin{figure*}[htbp]
\centering
\includegraphics[clip,scale=0.6]{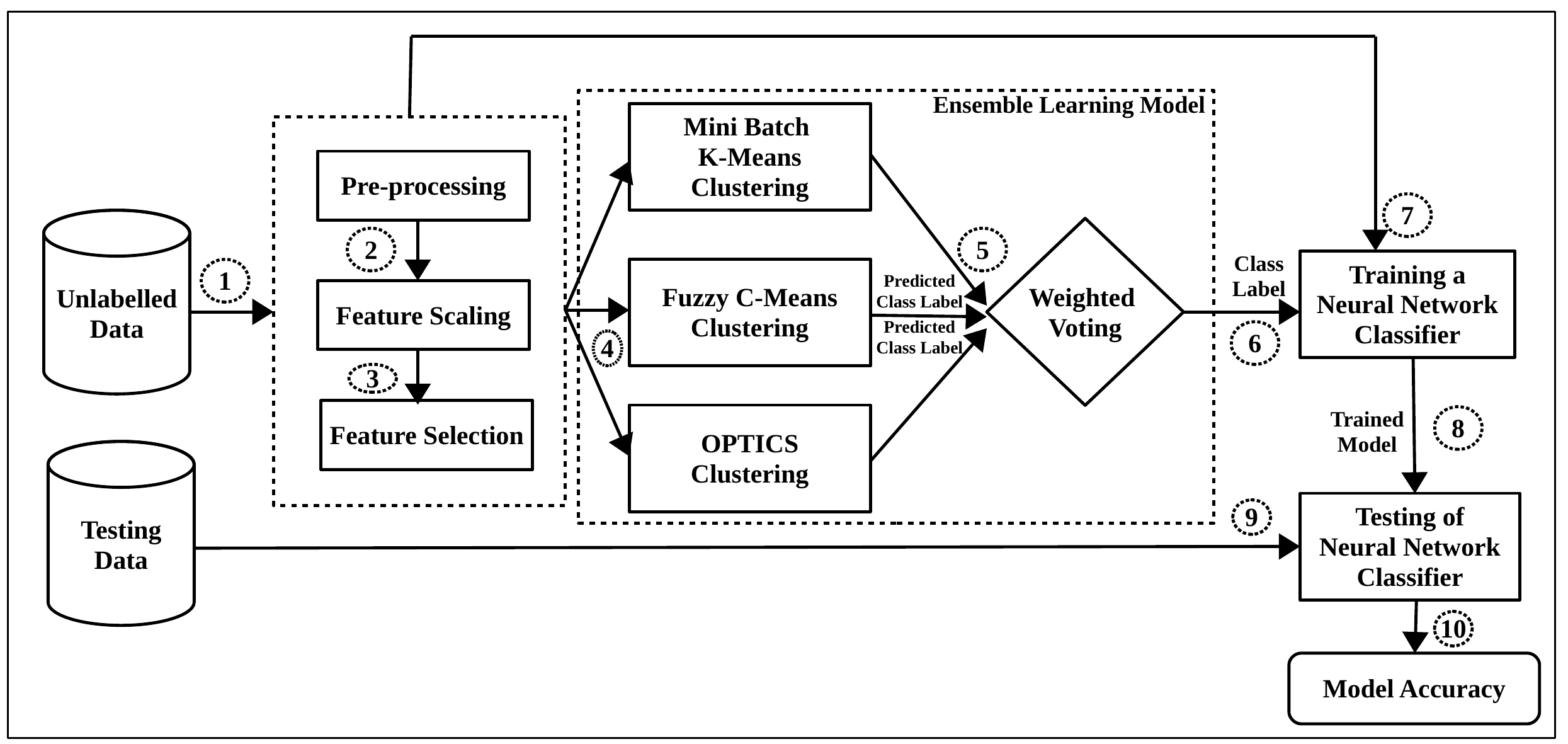}
\caption{Proposed Model}
\label{systemmodel}
\end{figure*}

\begin{table}[p]
\centering
\caption{Nomenclature} 
\begin{tabular}
{p{0.1\linewidth}p{0.4\linewidth}}\hline

Symbol &
Definition\\\hline

$\beta$ &
Regression coefficient\\

$\hat{\beta}$ &
LASSO regression coefficient\\

$\alpha$ &
Regularization parameter\\

$cof(i)$ &
Correlation coefficient for each feature vector $i$\\

$cov(X, y)$ &
Covariance between feature vector $X$ and class label $y$\\

$var(X)$ &
Variance of feature vector $X$\\

$P_i$ &
Predicted output from clustering algorithm $i$\\

$W_i$ &
Weights associated with the prediction $P_i$\\

$V$ &
Combined predicted class label\\

$\hat{V}$ &
Final predicted class label (after weighted voting)\\

$D_{UL}$ &
Unlabelled Dataset\\

$D_L$ &
Labelled Dataset\\

$FS$ &
Selected relevant features from the dataset\\

$\delta$ &
Threshold for correlation coefficient\\

\hline

\end{tabular}
\label{nomenclutureTable} 
\end{table}

\subsection{Feature Selection}
The performance of any machine learning algorithm depends on the selection of appropriate features from the given dataset. Selecting relevant features from a dataset helps to overcome overfitting or under-fitting of machine learning models. In the proposed model, we have used a combination of Correlation-based feature selection with LASSO regression to identify relevant features. In Correlation-based feature selection we select those features which are highly correlated with the class label. We have used Pearson's correlation coefficient \cite{benesty2009pearson} to find the correlation between feature vectors and the class label. The Pearson's correlation coefficient for each feature vector $i$ with the class label ($L$) is calculated using:

\begin{equation}
\label{correlationEq}
cof(i) = \frac{cov(X_i,L)}{\sqrt{(var(X_i)*var(L))}}
\end{equation}

where $X_i$ is the value of $i^{th}$ feature vector, $L$ is the class labels, $cov()$ and $var()$ are the covariance and the variance respectively. For selecting the relevant features, we calculate the absolute value of $cof(i)$ for each feature vector, and then select those features which have $cof(i)$ greater than threshold ($\delta$). In our study, after various iteration of proposed model for varying values of $\delta$, we got the optimal value of $\delta = 0.4$. The nomenclature for the symbols used in the paper are shown in table \ref{nomenclutureTable}.

In LASSO (Least Absolute Shrinkage and Selection Operator) regression some model coefficients are narrowed, while others are set to zero \cite{tibshirani1996regression}. It is this feature which makes it possible to select relevant features from a dataset. So, if $x_i = (x_{i1}, x_{i2}, \dots , x_{ip})^T$ represents the features selected from correlation method and $y_i$ represent the corresponding class label, then the LASSO regression coefficients are calculated as:

\begin{equation}
\label{LassoEq}
\hat{\beta}(i) =  \left[ \sum_{i=1}^{N} \left( y_i - \sum_{j} \beta_j x_{ij}\right)^2 \right] + \alpha \sum_{j} |\beta_j|
\end{equation}

where $\alpha$ is the constant that determines how much regularisation is applied. For large values of $\alpha$, the penalty term $\sum_{j} |\beta_j|$ has the effect of forcing some of the coefficients to absolutely zero, thus assisting in relevant variable selection. The optimal value of $\alpha$ was calculated using 10-fold cross validation with varying value of $\alpha$ from 0.1 to 8 with step-size of $0.01$. The complete feature selection mechanism is described in algorithm \ref{FeatureSelAlgo}.

\begin{algorithm}
\caption{Feature Selection Mechanism}
\label{FeatureSelAlgo}
\begin{algorithmic}[1]
\State \textbf{Input:} D($F_1,F_2,\dots, F_N,L$)
\State \textbf{Begin}
\State Create empty lists SF1 and SF2
\For{each feature-vector $V_i$ in D}
	\State calculate $cof(i)$ using equation \ref{correlationEq}
	\If{abs($cof(i))>\delta$}
		\State insert $V_i$ in SF1
	\EndIf
\EndFor
\For{each feature-vector $V_i$ in D}
	\State find $\hat{\beta}(i)$ using equation \ref{LassoEq}
	\If{$\hat{\beta}(i) \not = 0$}
		\State insert $V_i$ in SF2
	\EndIf
\EndFor
\State return (SF1 $\cup$ SF2)
\end{algorithmic}
\end{algorithm}

\begin{prop}
\label{pro1}
Time complexity of algorithm \ref{FeatureSelAlgo} is $\mathcal{O}(n)$, where n is the number of observations in each feature of the dataset.
\end{prop}

\begin{proof}
Step-3 of the algorithm takes $\mathcal{O}(1)$ time.

Step-4 to step-9 computes the Pearson's correlation coefficient for each pair of feature vector $V_i$ (with $n$ number of observations) and class label $L$. If $V$ is the number of features in a given dataset, then the total computations performed in Step-4 to step-9 are  $(n*V)$.

Step-10 to step-15 computes the LASSO regression coefficient for each feature vector $V_i$. Number of computations performed in Step-10 to step-15 are $(n*V^2)$ \cite{tibshirani1996regression}.

Finally, step-16 calculates union of two sets which takes $\mathcal{O}(V'+V'')$, where $V'$ is the number of features in set SF1 and $V''$ is the number of features in set SF2.

So, the overall time complexity of the algorithm is:

$\mathcal{O}(1 + (n*V) + (n*V^2) + V' + V'') = \mathcal{O}(n*V^2)$

Since, the number of features in a dataset are limited and small, hence $V$ can be neglected. So, the overall time complexity of algorithm \ref{FeatureSelAlgo} is: $\mathcal{O}(n)$.
\end{proof}

\subsection{Unsupervised Ensemble Learning Model with deep learning}
Ensemble learning is based on the principle that the combination of outputs from various learning models can produce more accurate results. The ensemble learning model can be implemented in two ways to produce multiple predicted results: Independent ensemble construction and Coordinated ensemble construction \cite{dietterich2002ensemble}. In independent ensemble construction approach, a learning algorithm can be executed independently multiple times on various training data-subsets or different learning models can be executed independently on same dataset, to produce multiple results which can then be combined using ensemble technique. Whereas, in coordinated ensemble construction the output of one learning algorithm can be used as an input to another learning algorithm, thus making all the base learning algorithms dependent on each other. 

In the proposed model, we have used an independent ensemble construction approach and used weighted voting for combining the output from different base learning models. The main purpose of the proposed ensemble learning model is to predict a class label for each data vector in the given unlabelled dataset. It is because of this we have used clustering techniques for predicting the class labels for data vectors. In the proposed model, we have used Mini Batch K-Means \cite{bottou1995convergence}, Fuzzy C-Means \cite{bezdek1984fcm} and OPTICS (Ordering Points to Identify the Clustering Structure) clustering \cite{ankerst1999optics} as the base learning models. The predicted output of each clustering technique will be either 0 or 1, where 0 represents non-malicious traffic (cluster-1) and 1 represents malicious traffic (cluster-2). The predicted output from each clustering algorithm, for each data entry, is combined using weighted voting using equation \ref{votingEq} to create two clusters -- one containing benign data and another containing malicious data.

An independent ensemble construction approach has been used because it helps to enhance the overall accuracy, reduce variance and produces a more stable prediction for the data by combining the results of different clustering algorithm used through weighted voting technique \cite{kim2003constructing}.

\begin{equation}
\label{votingEq}
V = \sum_{i=1}^{3} (P_i * W_i)
\end{equation}
Where, $W_i$ represent the weights associated with the prediction $P_i$ from each base clustering technique. Then the final predicted class label $\hat{V}$ is calculated as:

\begin{equation}
\hat{V} =
\bigg\{
	\begin{tabular}{cc}
		1  & if $V > 0.5$\\
		0  & otherwise
	\end{tabular}
\end{equation} 

After thorough performance analysis of the clustering technique used in our proposed model, the weights associated with the predicted value from Mini Batch K-Means and OPTICS clustering were set to 0.25 each and for Fuzzy C-Means it was set to 0.5. The entire process of converting an unlabelled dataset into a labelled dataset using the ensemble learning model is shown in algorithm \ref{EnsembleAlgo}.

\begin{algorithm}
\caption{Working of Ensemble Learning Model}
\label{EnsembleAlgo}
\begin{algorithmic}[1]
\State \textbf{Input:} 
\State $D_{UL}$: Unlabelled Dataset
\State $FS$: Feature set from algorithm \ref{FeatureSelAlgo}
\State \textbf{Begin}
\State Create an empty list $D_{L}$
\State Set $W_1=W_2 = 0.25$ \& $W_3 = 0.50$
\For{each data-entry $d_{UL}$ in $D_{UL}$}
	\State $P_1 = MBKmeans(FS(d_{UL}))$ 
	\State $P_2 = OPTICS(FS(d_{UL}))$
	\State $P_3 = FCmeans(FS(d_{UL}))$
	\State Cal. $V$ using eq. \ref{votingEq}
	\If{$V > 0.5$}
		\State Set $\hat{V} = 1$
	\Else
		\State Set $\hat{V} = 0$
	\EndIf
	\State Append $(d_{UL},\hat{V})$ in $D_{L}$
\EndFor
\State return $D_{L}$: Labelled Dataset
\end{algorithmic}
\end{algorithm}

\begin{prop}
\label{pro2}
Time complexity of algorithm \ref{EnsembleAlgo} is $\mathcal{O}(n^2)$, where n is the number of observations in the dataset.
\end{prop}

\begin{proof}
For each step-5 and step-6, it takes $\mathcal{O}(1)$ time for execution.

Step-7 to step-18 computes $P_1, P_2, P_3$ for each data entry of the dataset with $V_{FS}$ number of features (selected by algorithm \ref{FeatureSelAlgo}), $n$ number of observations in the dataset and $i$ number of iterations.

So, to compute $P_1$, it takes $\mathcal{O}(n*i*k*V_{FS})$ \cite{bottou1995convergence}, where $k$ is the number of clusters. Since, $k=2$ and $V_{FS}$ is small, so we can neglect them while calculating complexity. Hence, the time complexity of computing $P_1$ for entire dataset will be $\mathcal{O}(n*i)$.

Also, to compute $P_2$, it takes $\mathcal{O}(n^2)$ \cite{ankerst1999optics}, and

To compute $P_3$, it takes $\mathcal{O}(n*i*V_{FS}*c^2)$ \cite{bezdek1984fcm}, where $c$ is the number of clusters. Here again, $c=2$ and $V_{FS}$ is small, so we can neglect them while calculating complexity. Hence, the time complexity of computing $P_3$ for entire dataset will be $\mathcal{O}(n*i)$.

Step-12 to 16 takes $\mathcal{O}(n)$, and step-17 also takes $\mathcal{O}(n)$ time for execution.

Step-19 has $\mathcal{O}(1)$ time complexity.

Hence, the overall time complexity of algorithm \ref{EnsembleAlgo} is:

$\mathcal{O}(1 + (n*i) + n^2 + (n*i) + (2*n) + 1) = \mathcal{O}(n^2)$

\end{proof}

The proposition \ref{pro1} and \ref{pro2} suggest that both the proposed feature selection mechanism and ensemble learning model have polynomial time complexities, thus making them computationally feasible. 

The proposed ensemble model generates a labelled dataset. The generated labelled dataset is then used to train different deep learning models. In our study, we have used LSTM network \cite{kim2016long}, MLP \cite{friedman2017elements} and DBN \cite{hinton2006fast}, and using performance analysis an efficient model is selected for detecting the malicious attacks in the IoT network. The detailed architectures used for building each deep neural network model is shown in figure \ref{NNarch}. The efficient trained deep learning model can be then deployed at fog layer (in a fog computing architecture for IoT \cite{hu2017survey}), and periodically updated on cloud layer, to detect new/unknown network attacks in IoT devices at edge layer.

\begin{figure*}[htbp]
\centering
\includegraphics[clip,scale=0.55]{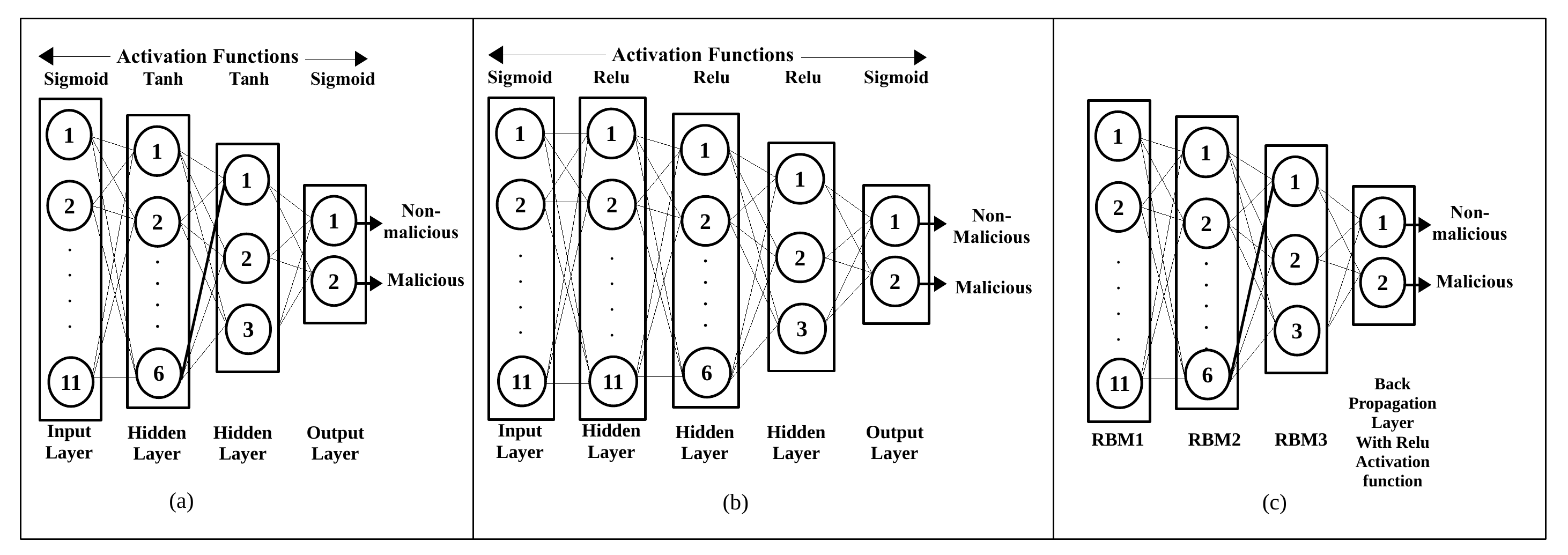}
\caption{Neural network architecture (a) LSTM NN (b) Multilayer Perceptron NN (c) DBN}
\label{NNarch}
\end{figure*}

\section{Results and Discussion}
\label{ResultsSection}
To evaluate the performance of proposed model for detecting attacks in IoT network, we used NSL-KDD \cite{dataset} \& TON-IoT \cite{TONdataset} Network dataset and implemented the model using Python programming language (ver. 3.7.3) with Scikit-learn and numpy packages. For implementing the deep learning model, we used Keras open-source software library. The literature suggests that NSL-KDD dataset has been used by many researcher to access the performance of attack detection models for IoT network \cite{gunupudi2017clapp,su2020bat,de2020hybrid}. Also, TON-IoT network dataset represents true network traffic of an IoT network and contains more IoT network specific network attacks, as the dataset was generated by real world heterogeneous IoT devices.

\subsection{Dataset Pre-processing}
NSL-KDD is an enhanced version of KDD CUP 99 dataset (initially generated in 1998). The enhancement took place in 2009 by Tavallaee et al. \cite{tavallaee2009detailed} by performing partitioning and removing the redundant records from the original dataset. The dataset is publicly available on the internet. In our study, we have used the NSL-KDD training data file, which consists of 125,973 network traffic samples with 41 features and classifies each traffic detail by the class label into malicious (represented by 1) or non-malicious (represented by 0). The malicious entries represent the network traffic generated by attackers. In the dataset the malicious entries represent DoS, Probe, R2L and U2R attacks. Before implementing the proposed model on the dataset, we have converted the data-type of some of the features, which include: Protocol type, Service and Flag. Mainly there are 3 types of protocols used in the dataset and we have assigned a numeric value to each: 1 for Internet Control Message Protocol, 2 for Transmission Control Protocol and 3 for User Datagram Protocol. Similarly, the 70 services in the service feature are also represented by numeric values from 1 to 70. The Flag attribute has 11 different values and each is represented by a numeric value as: OTH (1), REJ (2), RSTO (3), RSTOS0 (4), RSTR (5), S0 (6), S1 (7), S2 (8), S3 (9), SF (10) and SH (11).

TON-IoT dataset represents a combination of Telemetry, Operating and Network datasets, created at the IoT lab of UNSW Canberra in Australian Defence Force Academy (ADFA). In our study, we have used TON-IoT network dataset, which consists of 461,043 observations regarding network traffic generated by different IoT devices (malicious as well as non-malicious). It contains IoT network specific attacks like, DDoS, Ransomware attack, Injection attack, Man-in-the-middle attack etc. As described for NSL-KDD dataset, similar data pre-processing was carried out in TON-IoT network dataset for the features having string type data. Also, for feature scaling we have used standard scalar for both the datasets.

Once the dataset is ready to be used for implementing the proposed model, we first implement the proposed feature selection mechanism on the dataset. The proposed feature selection mechanism selects 11 \& 9 relevant features out of 41 \& 43 available features of NSL-KDD \& TON-IoT network datasets respectively. This results in the feature reduction rate (FRR) of 73.17\% \& 79.07\% for respective datasets, which can be calculated using equation \ref{FRR}.

\begin{equation}
\label{FRR}
FRR = 1 - \frac{n(SF)}{n(F)}
\end{equation}

Where, $n(SF)$ represents the number of selected features and $n(F)$ represent the total number of features in the dataset. Also, the complete dataset features and selected features are shown in figure \ref{datasetFeatures} \& \ref{datasetFeaturesTON}. Network attacks in both the datasets are shown in figure \ref{attacksFigure}.

\begin{figure}[h]
\centering
  \begin{subfigure}[b]{0.49\linewidth}
    \includegraphics[clip, scale=0.4]{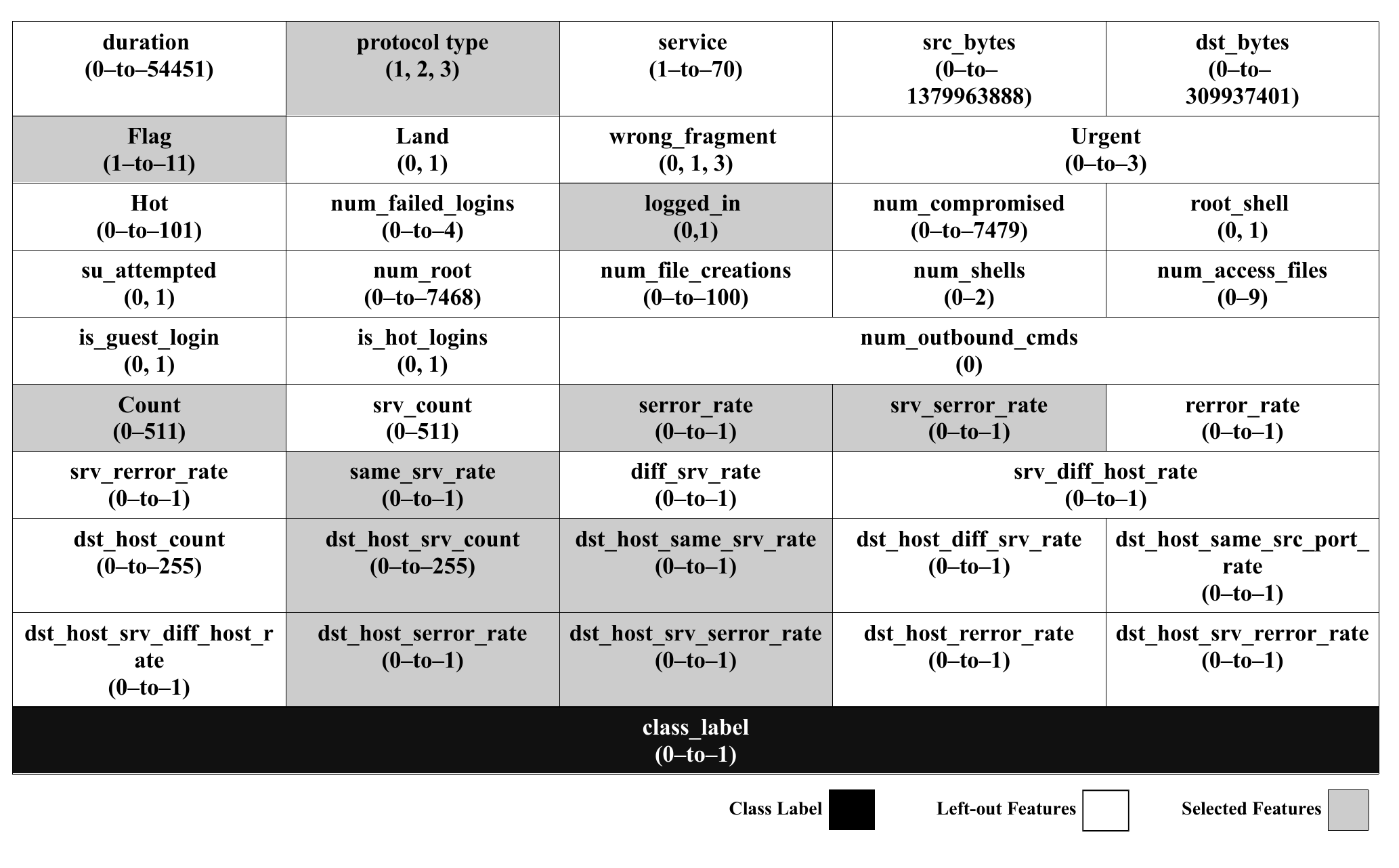}
    \caption{NSL-KDD dataset}\label{datasetFeatures}
  \end{subfigure}
  \begin{subfigure}[b]{0.49\linewidth}
    \includegraphics[clip, scale=0.4]{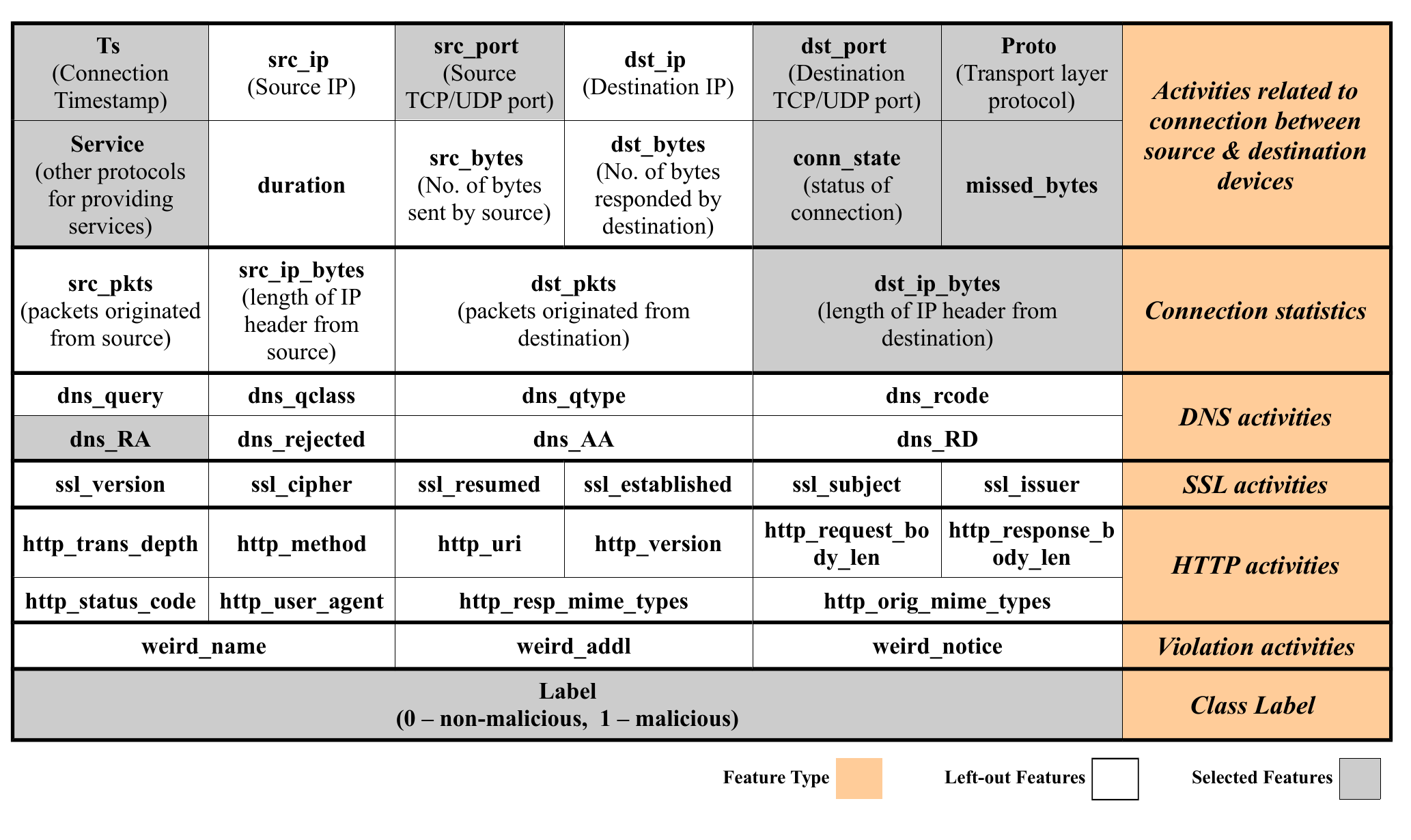}
    \caption{TON-IoT network dataset}\label{datasetFeaturesTON}
  \end{subfigure}
  \caption{Dataset Features}
\end{figure}

\begin{figure}[h]
\centering
\includegraphics[clip,scale=0.35]{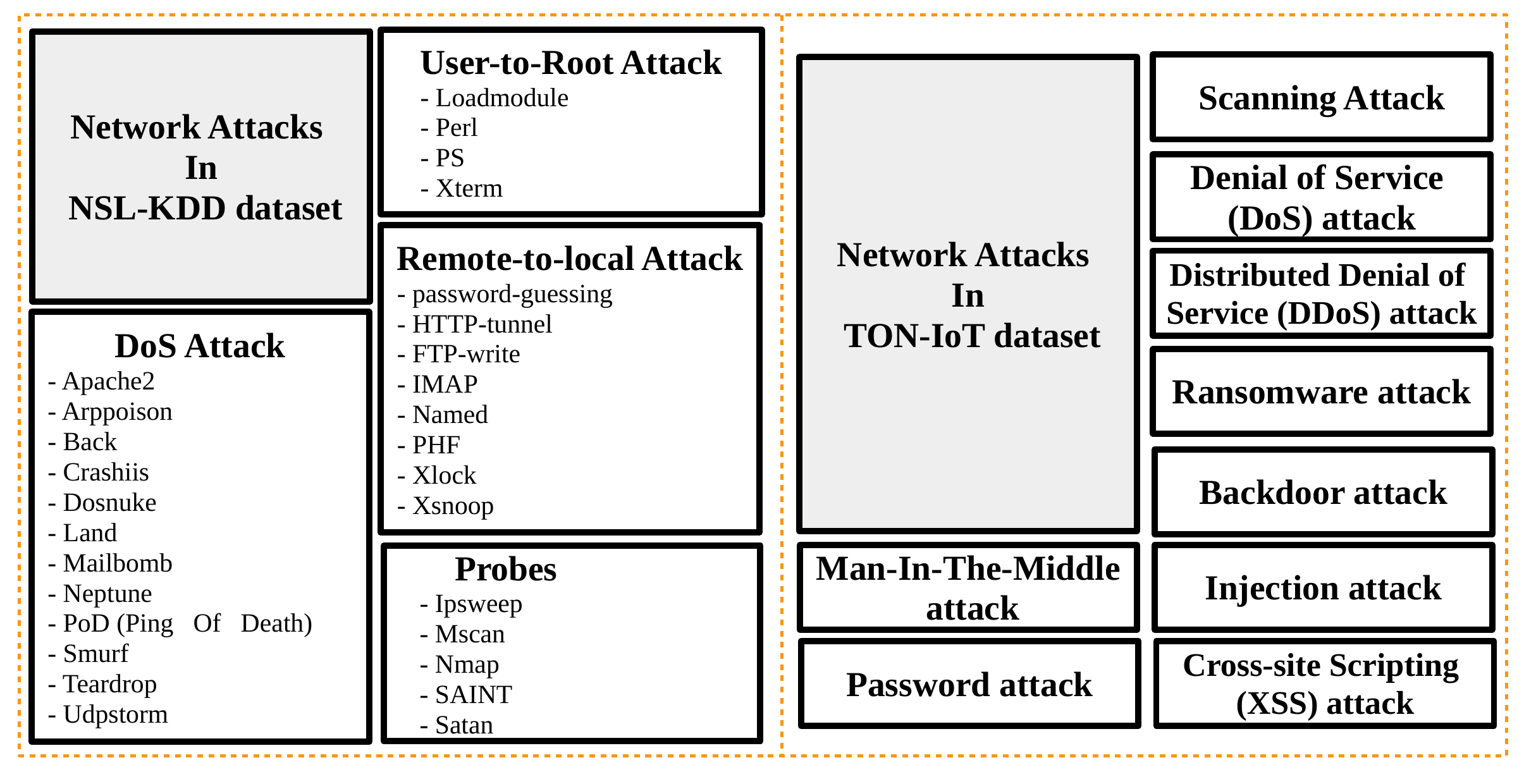}
\caption{Network attacks in NSL-KDD and TON-IoT network Datasets}
\label{attacksFigure}
\end{figure}

\subsection{Evaluation metrics}

\begin{table}[p]
\centering
\caption{Confusion Matrix} 
\begin{center}
\begin{tabular}{c|cc} \hline

\multirow{2}{*}{\textbf{Actual Value}} & \multicolumn{2}{c}{\textbf{Predicted Value}}\\
& \textbf{Malicious} & \textbf{Non-malicious}\\\hline
\textbf{Malicious} & True Positive (TP) & False Negative (FN)\\
\textbf{Non-malicious} & False Positive (FP) & True Negative (TN)\\\hline

\end{tabular}
\end{center} 
\label{ConfusionMatrix} 
\end{table}

In the proposed model we have used clustering as well as classification techniques and each produces binary output (0 or 1) representing non-malicious or malicious traffic. So, to measure the performance of various techniques used in the model, we have used confusion matrix (as shown in table \ref{ConfusionMatrix}) and following evaluation metrics:

\begin{itemize}
\item Homogeneity $(h)$ refers to how identical the samples in a cluster are.
\begin{equation}
h = 1 - \frac{H(Y_{pred.}|Y_{true})}{H(Y_{true})}
\end{equation}
\item The clustering algorithm's completeness $(c)$ determines how many related samples are grouped together.
\begin{equation}
c = 1 - \frac{H(Y_{pred.}|Y_{true})}{H(Y_{pred.})}
\end{equation}
Where, $H()$ represents the entropy.
\item V-measure $(V)$ is the harmonic means between $h$ and $c$.
\begin{equation}
V = \frac{2*h*c}{h+c}
\end{equation}
\item A Rand Index $RI$ evaluates all pairs of samples from true and predicted clusterings and finds the similarity between the two. 
\begin{equation}
RI = \frac{TP+TN}{TP+FN+FP+TN}
\end{equation}
The Adjusted Rand Index $(ARI)$ is then calculated as:
\begin{equation}
ARI = \frac{(RI - ExpectedRI)}{(max(RI) - ExpectedRI)} 
\end{equation}
\item If $X$ and $Y$ are two clusterings, then Adjusted Mutual Information $(AMI)$ between two clusterings is calculated as:
\begin{equation}
\hspace{-3em}
AMI(X, Y) = \frac{MI(X, Y) - E(MI(X, Y))}{avg(H(X), H(Y)) - E(MI(X, Y))}
\end{equation}
\item Out of all the samples analysed, the accuracy $Acc$ indicates the percentage of samples that are correctly identified.
\begin{equation}
Acc = \frac{TP+TN}{TP+FN+FP+TN}
\end{equation}
\item The capacity of the classifier to detect only relevant data is referred to as precision.
\begin{equation}
Precision= \frac{TP}{TP+FP}
\end{equation}
\item The ratio of total positives recognised by the system to real positives across the entire system is given by Recall.
\begin{equation}
Recall = \frac{TP}{TP+FN}
\end{equation}
\item F1-score represents the harmonic mean between precision and recall.
\begin{equation}
F1-score = \frac{2*TP}{2*TP+FP+FN}
\end{equation}
\item The False Alarm Rate $(FAR)$ is the percentage of non-malicious samples in the given dataset that were incorrectly categorised as an attack.
\begin{equation}
FAR = \frac{FP}{TN + FP}
\end{equation}
\item The rate of accurately predicted non-malicious samples for all possible non-malicious samples in the given dataset is represented by specificity.
\begin{equation}
Specificity = \frac{TN}{TN + FP}
\end{equation}
\item Matthews Correlation Coefficient $(MCC)$ finds the correlation between actual class label and the predicted class value.
\scriptsize
\begin{equation}
\hspace{-3em}
MCC = \frac{(TP*TN)-(FP*FN)}{\sqrt{(TP+FN)*(TP+FP)*(TN+FP)*(TN+FN)}}
\end{equation}

\end{itemize}

\subsection{Performance Analysis}
In the initial phase, we have implemented the clustering techniques (as discussed in the proposed model) on NSL-KDD and TON-IoT network datasets for selected features, to evaluate the clustering performance of each model. The entries of confusion matrices and other performance parameters for clustering the network traffic, in both the dataset, into malicious and non-malicious are shown in table \ref{clusteringTable}. Figure \ref{clusteringFigure} \& \ref{clusteringFigureTON}  shows the comparative analysis of the performance of each clustering technique used for both the datasets. It is evident from the results that Fuzzy C Means clustering outperforms the other, that is why we have set more weight on its predicted class label.

\begin{table*}[p]
\centering
\caption{Results of clustering techniques for predicting class label (malicious or non-malicious)} 
\footnotesize
\begin{center}
\begin{tabular}
{p{0.1\linewidth}|p{0.08\linewidth}p{0.08\linewidth}p{0.08\linewidth}|p{0.08\linewidth}p{0.08\linewidth}p{0.08\linewidth}}\hline

 & \multicolumn{3}{c|}{\textbf{For NSL-KDD Dataset}} & \multicolumn{3}{c}{\textbf{For TON-IoT Dataset}}\\

Performance Parameters &
Mini Batch K-Means &
OPTICS &
Fuzzy C-Means &
Mini Batch K-Means &
OPTICS &
Fuzzy C-Means\\\hline

TP &
40951 &
56711 &
43906 &
112483 &
155772 &
120600\\

FN &
17679 &
1919 &
14724 &
48560 &
5271 &
40443\\

FP &
1803 &
20187 &
176 &
8032 &
89929 &
784\\

TN &
65540 &
47156 &
67167 &
291968 &
210071 &
299216\\

Precision &
0.957 &
0.737 &
0.996 &
0.933 &
0.633 &
0.993\\

Recall &
0.698 &
0.967 &
0.748 &
0.698 &
0.967 &
0.748\\

Homogeneity &
0.426 &
0.398 &
0.543 &
0.482 &
0.398 &
0.543\\

Completeness &
0.427 &
0.411 &
0.58 &
0.463 &
0.411 &
0.58\\\hline

\end{tabular}
\end{center} 
\label{clusteringTable} 
\end{table*}

\begin{figure}[h]
\centering
  \begin{subfigure}[b]{0.49\linewidth}
    \includegraphics[clip, scale=0.3]{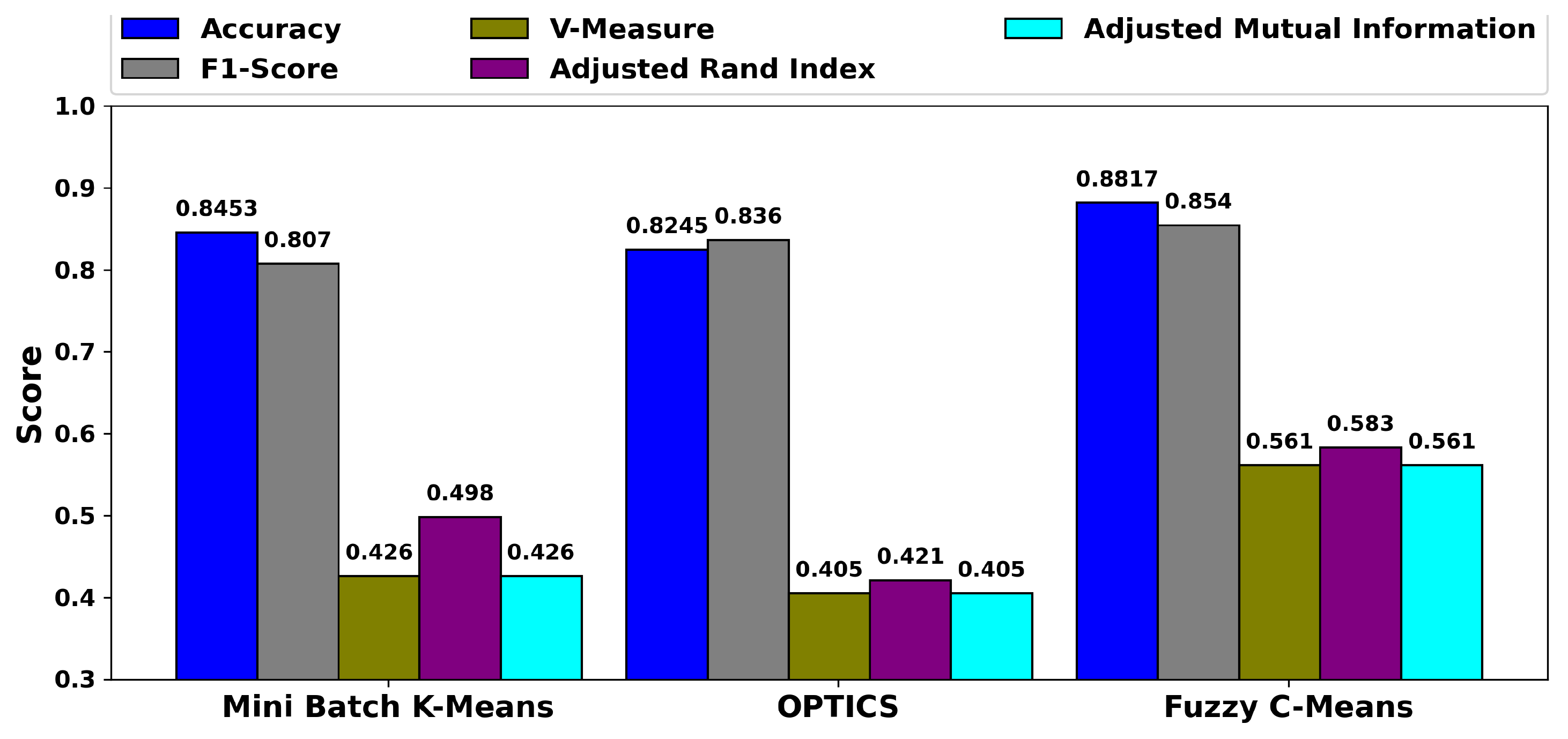}
    \caption{For  NSL-KDD dataset}\label{clusteringFigure}
  \end{subfigure}
  \begin{subfigure}[b]{0.49\linewidth}
    \includegraphics[clip, scale=0.3]{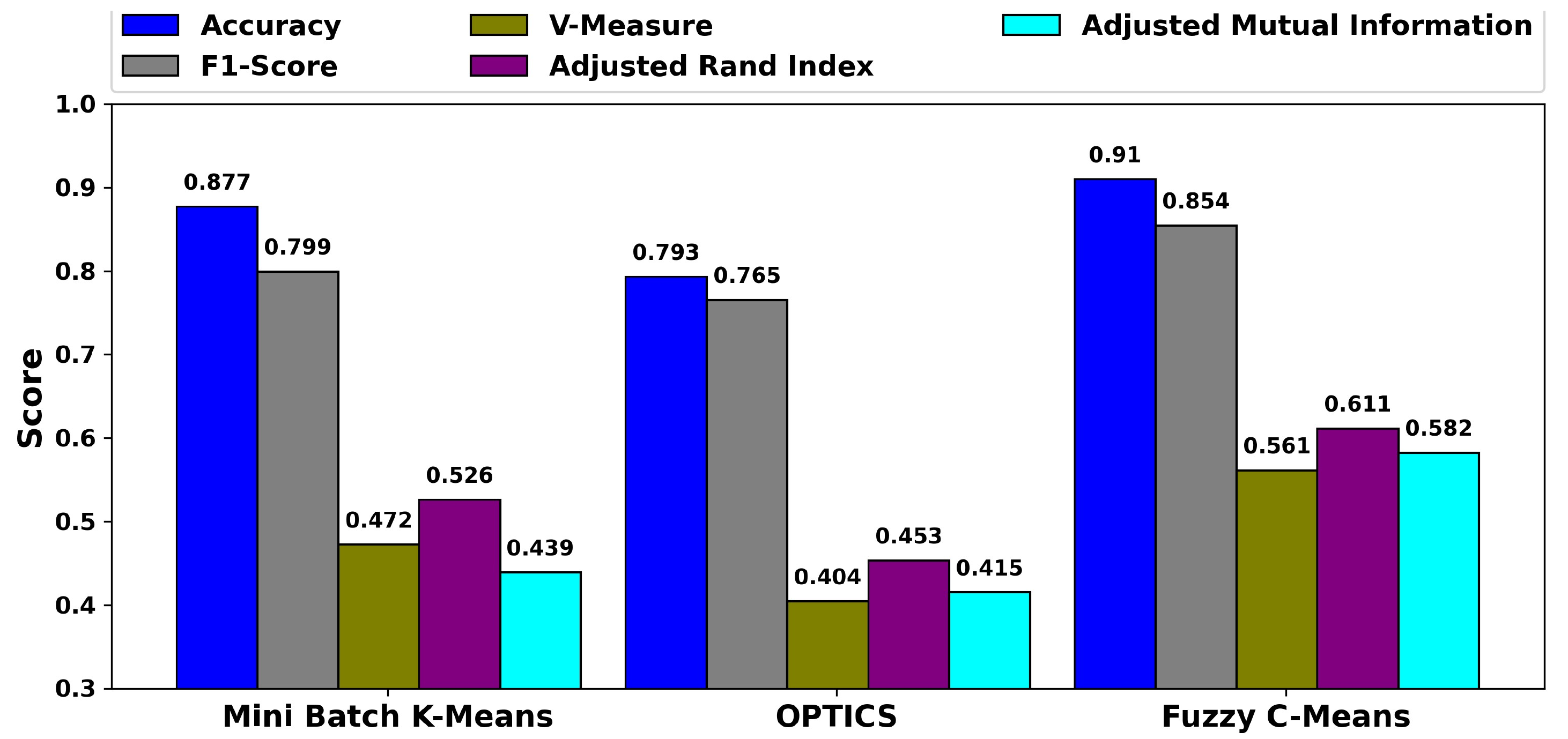}
    \caption{For  TON-IoT network dataset}\label{clusteringFigureTON}
  \end{subfigure}
  \caption{Comparative results of various clustering techniques used for predicting class label (malicious and non-malicious)}
\end{figure}

After finalising the weights for performing weighted voting in the ensemble learning model, the proposed ensemble model is now executed where the entire unlabelled NSL-KDD and TON-IoT network datasets separately are given to the clustering algorithms to predict the class label for each traffic observation. The predicted class labels are used as an input to the weighted voting process to predict the final class label for each observation. After, the entire unlabelled datasets are converted into labelled datasets, then these generated dataset are used to train various deep neural network models. The architectures used for building each deep neural network model is shown in figure \ref{NNarch}.

The optimal values of hyperparameters of various deep learning models were selected after iterative simulation of the models with varying values. For implementing LSTM deep neural networks, the value of various hyperparameters is given as: Learning rate: 0.2, decay: 0.01, Number of epochs: 40, Batch size: 250, Optimizer: Adam, and Dropout rate at each LSTM layer: 0.2. Similarly, for Multilayer perceptron, the value of various hyperparameters is given as: Learning rate: 0.001, Number of epochs: 40, Batch size: 200, Optimizer: stochastic gradient descent (SGD). All the layers of DBN are created using Restricted Boltzmann Machine (RBM), except the output layer (which uses back-propagation layer to generate the final output). The value of various hyperparameters in DBN are: Learning rate: 0.08, Decay: 0.006, Optimizer: Adamax, Number of epochs: 30, Batch size: 300. Figure \ref{lossCurve} shows the loss curve for various models during the training phase.

\begin{figure}[h]
\centering
\includegraphics[clip,scale=0.35]{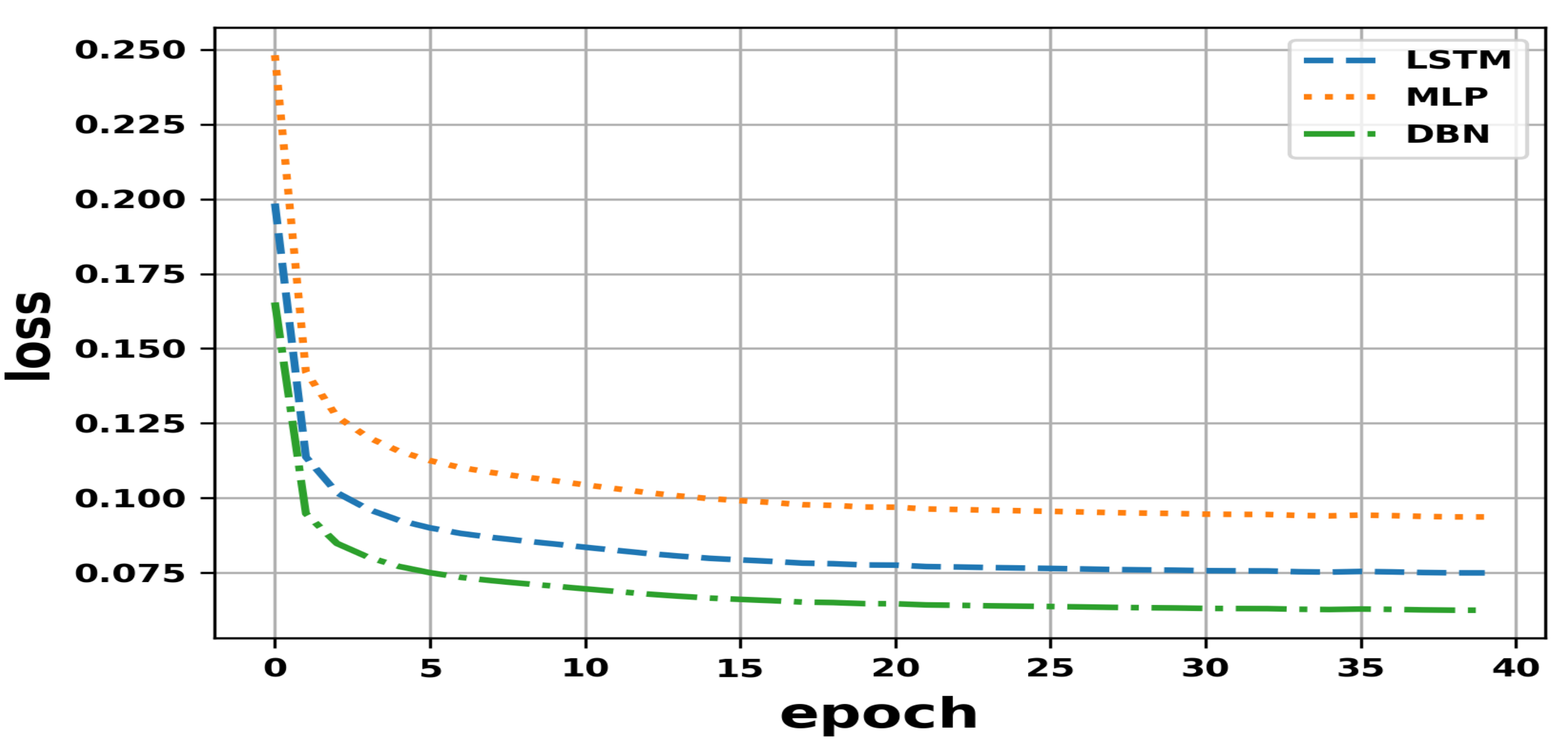}
\caption{Loss curve during training process for various deep learning models}
\label{lossCurve}
\end{figure}

Once, the deep learning models are trained using the labelled datasets generated by the proposed model, they are then tested on actual NSL-KDD and TON-IoT network datasets separately to detect malicious and non-malicious traffic. The performance of deep learning models are shown in table \ref{clasificationTable} and figure \ref{clacificationFigure} \& \ref{clacificationFigureTON}.

\begin{table}[p]
\centering
\caption{Results of various deep learning models for detecting malicious network traffic, after being trained on labelled dataset generated by proposed model} 
\footnotesize
\begin{center}
\begin{tabular}
{p{0.16\linewidth}|p{0.08\linewidth}p{0.08\linewidth}p{0.08\linewidth}|p{0.09\linewidth}p{0.09\linewidth}p{0.09\linewidth}}\hline

 & \multicolumn{3}{c|}{\textbf{For NSL-KDD Dataset}} & \multicolumn{3}{c}{\textbf{For TON-IoT Dataset}}\\
Parameters &
LSTM &
MLP &
DBN &
LSTM &
MLP &
DBN\\\hline

TP &
54505 &
55469 &
57253 &
149713 &
152360 &
157261\\

FN &
4125 &
3161 &	
1377 &
11330 &
8683 &
3782\\

FP &
2442 &
1293 &
1596 &
10879 &
5760 &
7110\\

TN &
64901 &
66050 &
65747 &
289121 &
294240 &
292890\\

Precision &
0.957 &
0.977 &
0.972 &
0.932 &
0.963 &
0.956\\

Recall &
0.929 &
0.946 &
0.976 &
0.929 &
0.946 &
0.976\\

FAR &
0.036 &
0.019 &
0.023 &
0.036 &
0.019 &
0.0237\\

Training Time (in Sec)&
11.84&
18.31&
15.63&
19.44&
28.64&
24.25\\

Testing Time (in Sec)&
0.023&
0.036&
0.029&
0.056&
0.074&
0.068\\\hline

\end{tabular}
\end{center} 
\label{clasificationTable} 
\end{table}

\begin{figure}[h]
\centering
  \begin{subfigure}[b]{0.49\linewidth}
    \includegraphics[clip, scale=0.3]{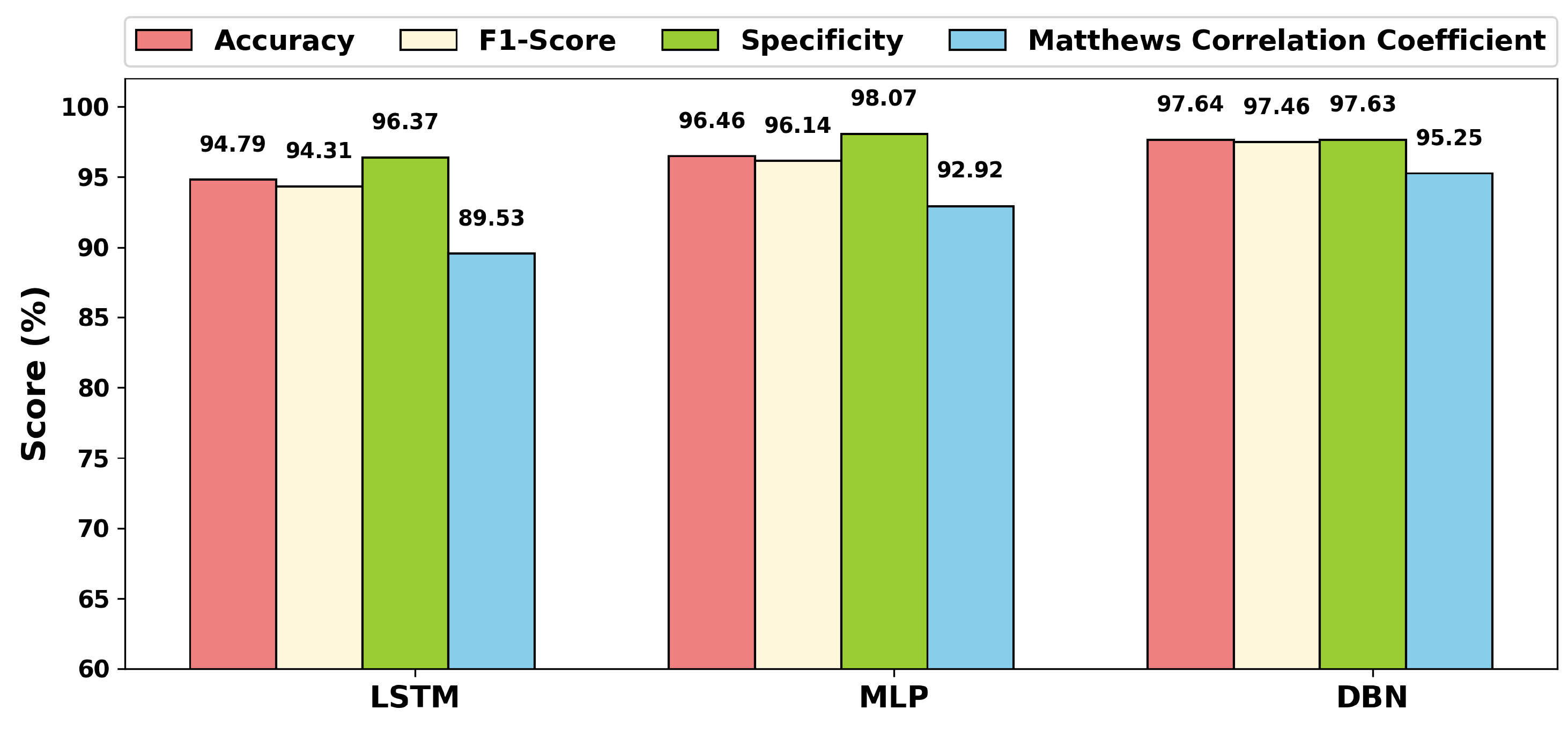}
    \caption{For  NSL-KDD dataset}\label{clacificationFigure}
  \end{subfigure}
  \begin{subfigure}[b]{0.49\linewidth}
    \includegraphics[clip, scale=0.3]{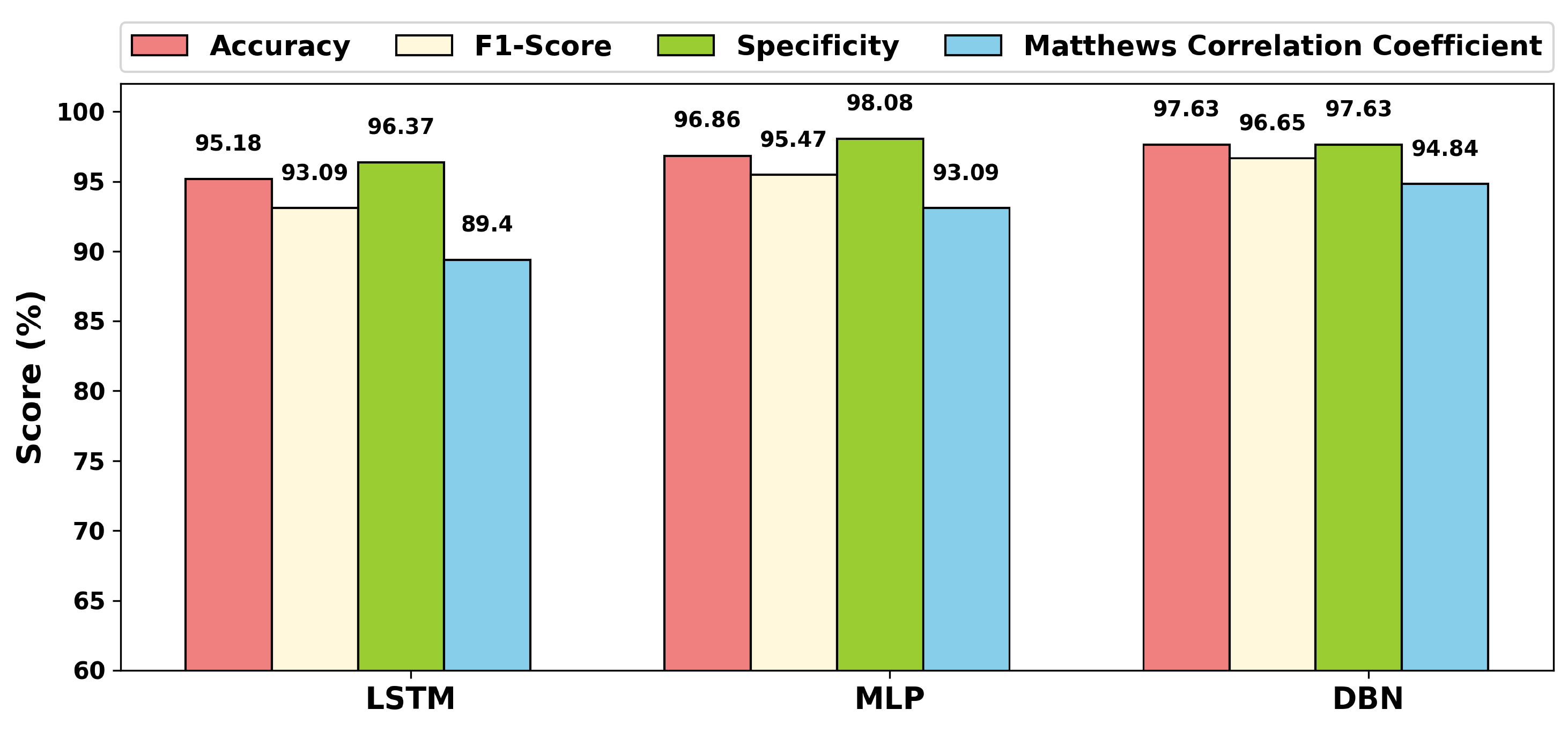}
    \caption{For  TON-IoT network dataset}\label{clacificationFigureTON}
  \end{subfigure}
  \caption{Comparative results of models for detecting malicious network traffic, after being trained on labelled dataset generated by proposed model}
\end{figure}

\begin{figure}[h]
\centering
  \begin{subfigure}[b]{0.45\linewidth}
    \includegraphics[scale=0.3]{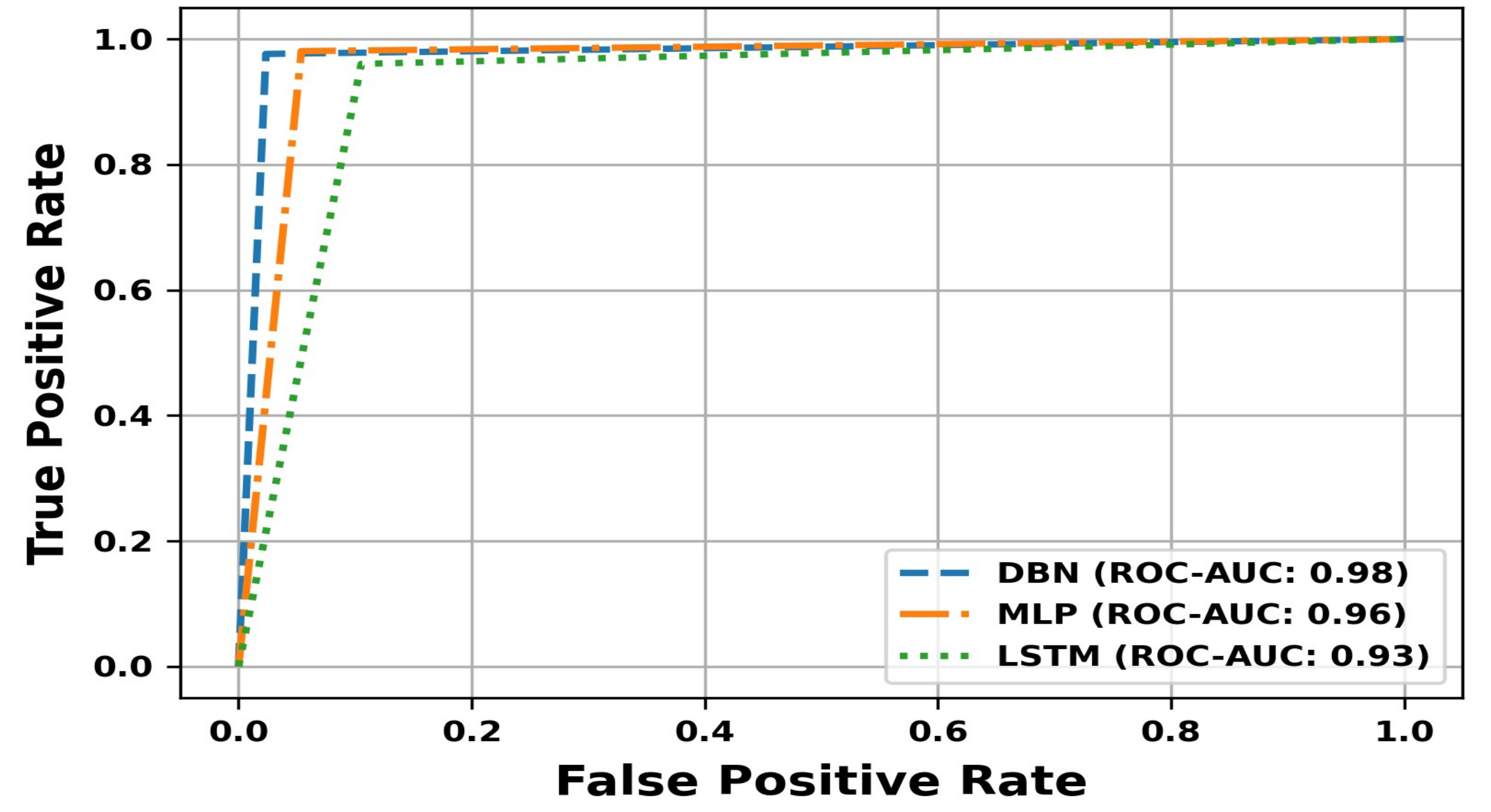}
    \caption{For  NSL-KDD dataset}\label{ROC}
  \end{subfigure}
  \begin{subfigure}[b]{0.45\linewidth}
    \includegraphics[scale=0.3]{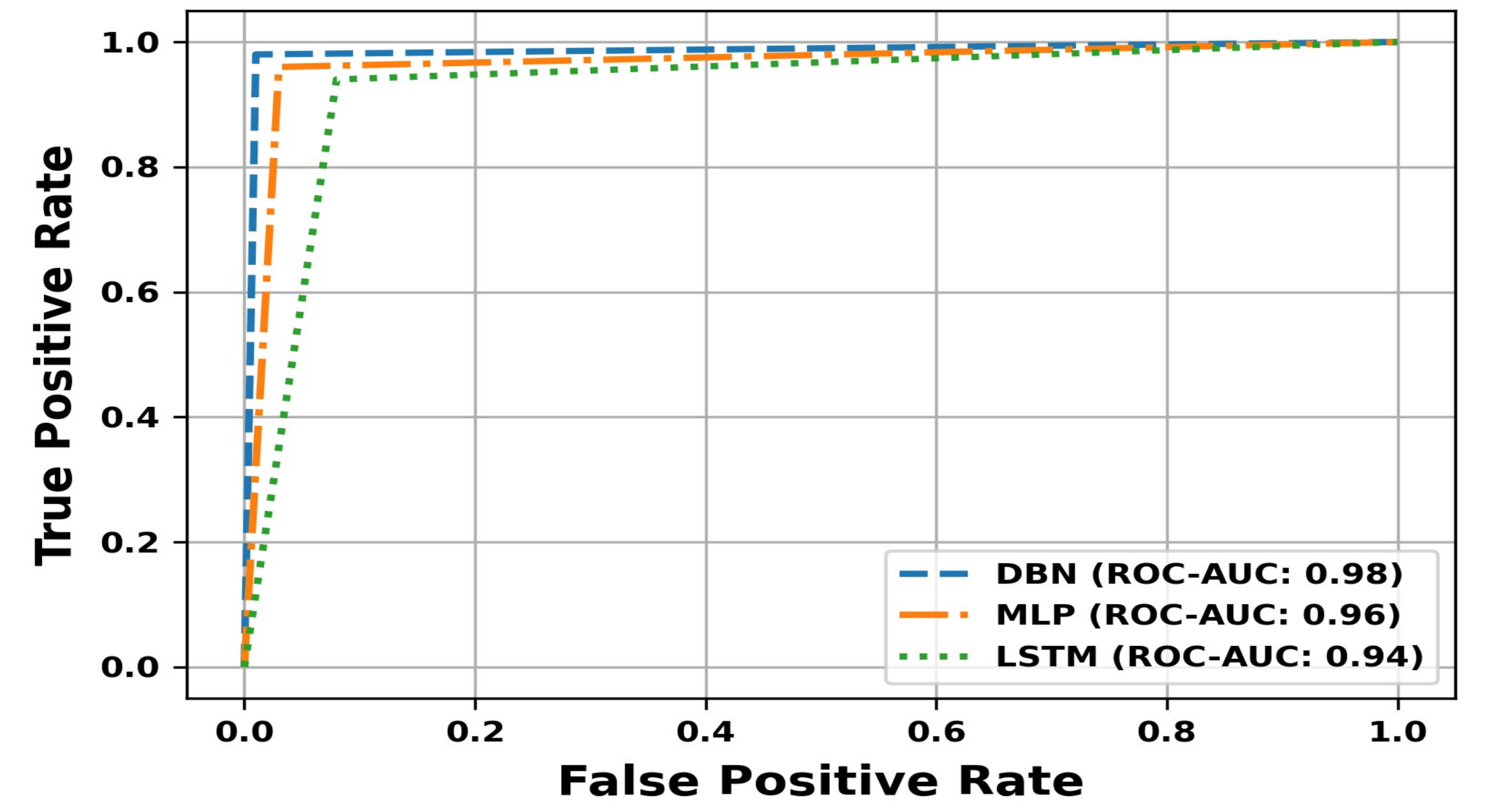}
    \caption{For  TON-IoT network dataset}\label{ROC_TON}
  \end{subfigure}
  \caption{ROC curves of models for detecting malicious  network traffic, after being trained on labelled dataset generated by proposed model}
\end{figure}

\begin{figure}
\centering
  \begin{subfigure}{.45\textwidth}
    \includegraphics[scale=0.3]{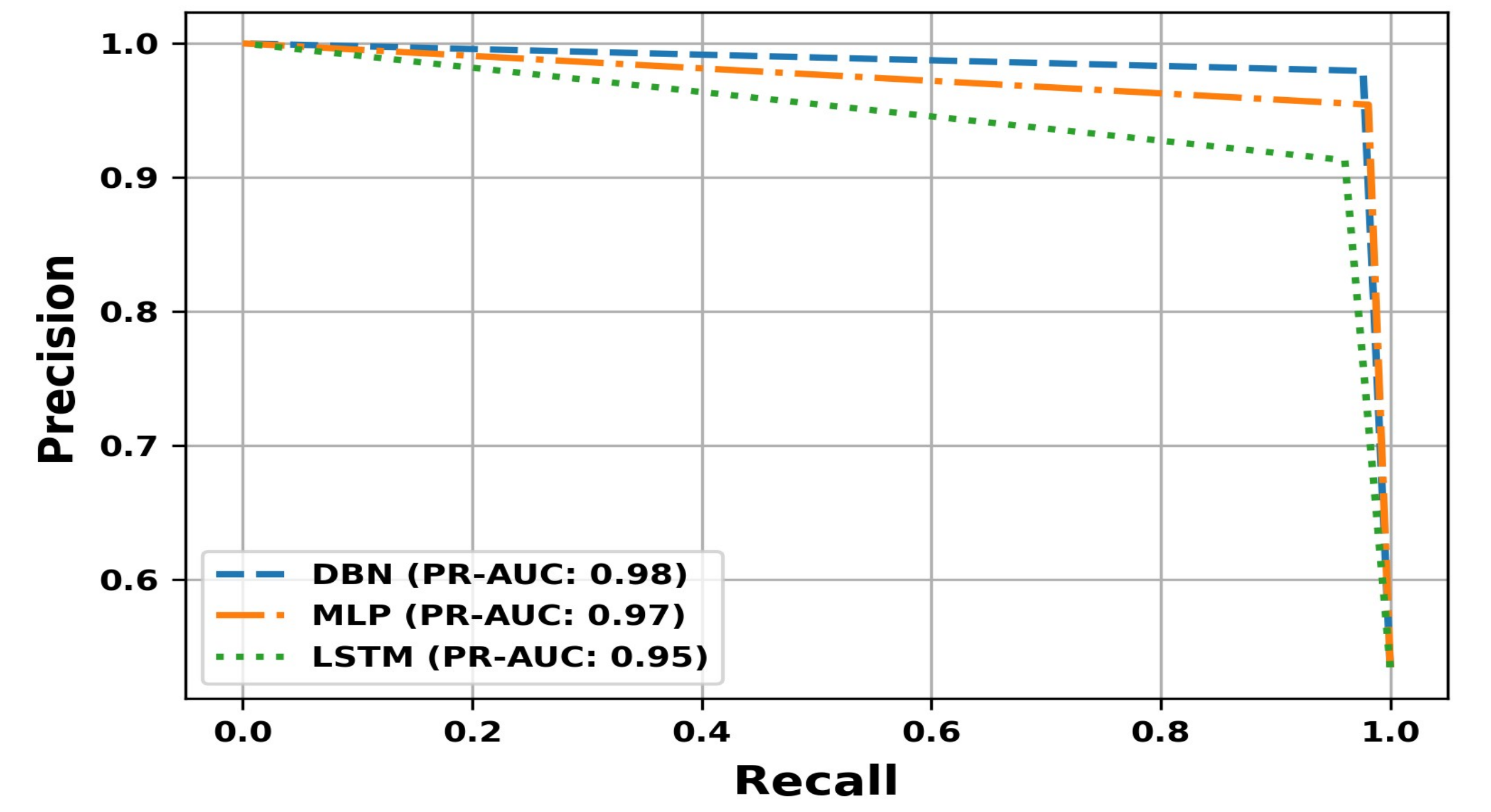}
    \caption{For NSL-KDD dataset}\label{PR}
  \end{subfigure}
  \begin{subfigure}{.45\textwidth}
    \includegraphics[scale=0.3]{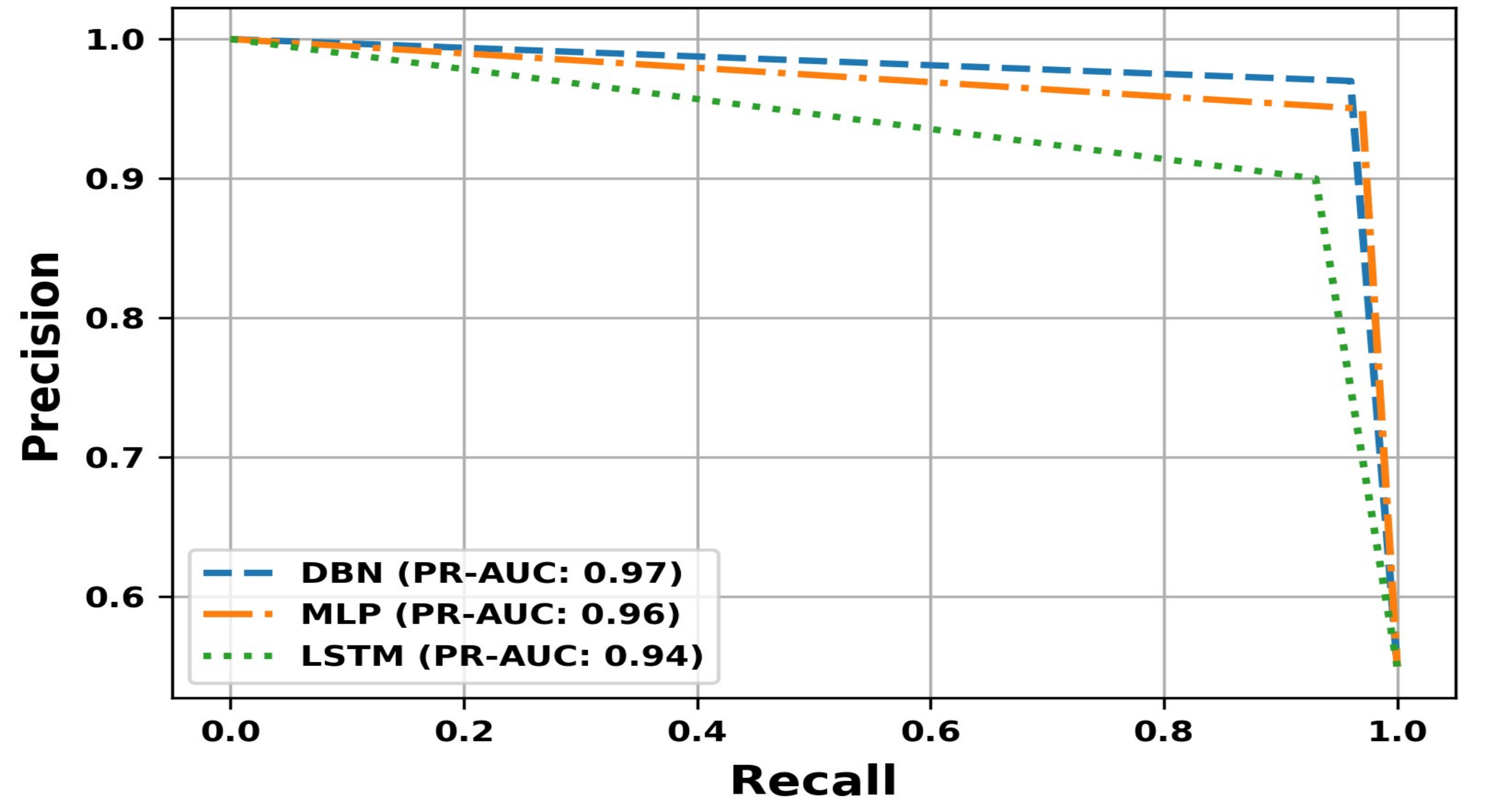}
    \caption{For TON-IoT network dataset}\label{PR_TON}
  \end{subfigure}
\caption{PR curve of models for detecting malicious network traffic, after being trained on labelled dataset generated by proposed model}
\end{figure}

Figure \ref{clacificationFigure} \& \ref{clacificationFigureTON} describe the comparison of results of various deep learning models, which were trained using NSL-KDD and TON-IoT network datasets (labelled by the proposed ensemble learning model). The results in the figure show that all the studied models produce acceptable results for detecting malicious attacks in IoT networks, this proves that the labelled dataset generated by the proposed ensemble learning model has effectively predicted the class label for each data entry in the unlabelled dataset. Also, it is evident from the testing results that DBN model outperforms the other with an attack detection accuracy of 97.6\% and has least false alarm rate of 2.3\% on both the datasets. 

Figure \ref{ROC} \& \ref{ROC_TON} show the Receiver Operating Characteristic (ROC) and figure \ref{PR} \& \ref{PR_TON} show the Precision-Recall (PR) curves for NSL-KDD \& TON-IoT network datasets respectively, with Area Under the Curve (AUC) values. The AUC value lies in the range of $[0,1]$ and a classifier whose AUC value is closer to 1 is considered efficient. Again, these figures also depict that DBN has better results with ROC-AUC  $= 0.98$ (for both the datasets) \& PR-AUC $= 0.98$ (for NSL-KDD) and $ 0.97$ (for TON-IoT network Dataset).

\begin{figure}[h]
\centering
\includegraphics[clip,scale=0.45]{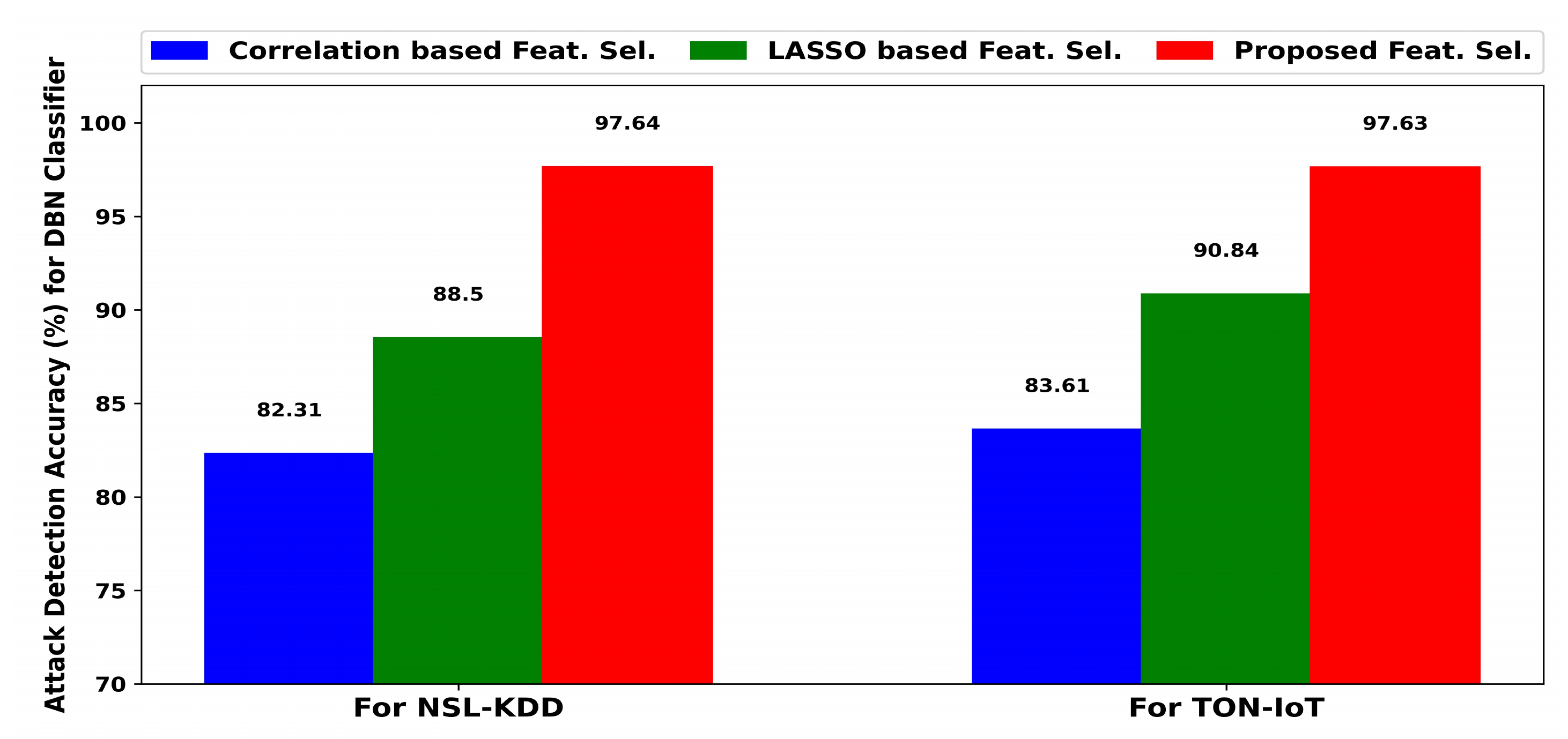}
\caption{Attack detection accuracy of DBN model, after being trained on labelled dataset generated by proposed model for the features selected by Correlation based, LASSO Reg. based and proposed feature selection mechanisms resp.}
\label{featSelSig}
\end{figure}

To show the significance of the proposed feature selection mechanism, which is an amalgamation of correlation based and LASSO reg. based feature selection mechanism, we have implemented three cases. In case-I, we have selected the features using correlation based feature selection mechanism (for NSL-KDD \& TON-IoT network datasets resp.), used proposed ensemble learning model to predict the class labels for malicious and non-malicious records, then train DBN model for detecting malicious attacks and note the detection accuracy. In case-II and III, again the same procedure is followed for LASSO regression based and proposed feature selection  mechanisms respectively. The attack detection accuracies for case-I, II \& III are shown in  figure \ref{featSelSig}. The study clearly indicates that the proposed feature selection helps to extract most relevant features from the datasets and thus helps to enhance attack detection accuracy.

\section{Conclusion \& Future Work}
\label{ConclusionSection}
In this paper, we proposed an unsupervised ensemble based learning model which is able to transform an unlabelled network dataset into labelled dataset, thereby predicting new/unknown attacks. The paper also proposed a feature selection mechanism which is able to discover relevant features in NSL-KDD \& TON-IoT network datasets for detecting malicious attacks, resulting in FRR of 73.17\% \& 79.07\% respectively. The proposed ensemble learning model first converts each unlabelled dataset into labelled dataset, then use it to train a deep learning model (separately for each dataset). In the study, we used LSTM, MLP \& DBN deep learning models and results show that DBN is more efficient than the other two studied deep learning models for detecting malicious attacks in an IoT network, with attack detection accuracy of 97.6\% and FAR of 2.3\% on both the studied datasets.

The proposed unsupervised ensemble based learning model can be implemented with fog computing architecture, where the model can be deployed at cloud layer to identify new/unknown attacks by analysing the unlabelled IoT network traffic. The network traffic labelled by the proposed model can then be used to train DBN deep learning model to detect network attack. The trained DBN can then be deployed at the fog layer to analyse the network traffic of edge devices and detect attacks, with periodic updates at the cloud layer for identifying new attacks. The use of fog computing architecture will decrease the load on fog as well as energy constrained IoT edge devices. In future, we will implement the model on a real world IoT network with fog computing architecture to further investigate the efficiency and complexity of the proposed model.

\section*{Acknowledgement}
We would like to thank TEQIP-III project and MITS, Gwalior for supporting this research.

\bibliographystyle{unsrt}  
\bibliography{CCPE_arxiv_mir_and_shah}

\begin{thebibliography}{10}
\providecommand \doibase [0]{http://dx.doi.org/}%

\bibitem{javed2018internet}
Javed F, Afzal MK, Sharif M, Kim BS. Internet of Things (IoT) operating systems
  support, networking technologies, applications, and challenges: A comparative
  review. {\it IEEE Communications Surveys \& Tutorials} 2018\string;
  20(3)\string: 2062--2100.

\bibitem{ciuonzo2021distributed}
Ciuonzo D, Rossi PS, Varshney PK. Distributed detection in wireless sensor
  networks under multiplicative fading via generalized score tests. {\it IEEE
  Internet of Things Journal} 2021\string; 8(11)\string: 9059--9071.

\bibitem{meneghello2019iot}
Meneghello F, Calore M, Zucchetto D, Polese M, Zanella A. IoT: Internet of
  threats? A survey of practical security vulnerabilities in real IoT devices.
  {\it IEEE Internet of Things Journal} 2019\string; 6(5)\string: 8182--8201.

\bibitem{hassan2019current}
Hassan WH, others . Current research on Internet of Things (IoT) security: A
  survey. {\it Computer networks} 2019\string; 148\string: 283--294.

\bibitem{xiao2018iot}
Xiao L, Wan X, Lu X, Zhang Y, Wu D. IoT security techniques based on machine
  learning: How do IoT devices use AI to enhance security?. {\it IEEE Signal
  Processing Magazine} 2018\string; 35(5)\string: 41--49.

\bibitem{deng2019mobile}
Deng L, Li D, Yao X, Cox D, Wang H. Mobile network intrusion detection for IoT
  system based on transfer learning algorithm. {\it Cluster Computing}
  2018\string; 22(4)\string: 9889--9904.

\bibitem{rathore2018semi}
Rathore S, Park JH. Semi-supervised learning based distributed attack detection
  framework for IoT. {\it Applied Soft Computing} 2018\string; 72\string:
  79--89.

\bibitem{sahay2018traffic}
Sahay R, Geethakumari G, Modugu K, Mitra B. Traffic convergence detection in
  IoT LLNs: a multilayer perceptron based mechanism. In: IEEE. ; 2018\string:
  1715--1722.

\bibitem{sahay2018efficient}
Sahay R, Geethakumari G, Mitra B, Sahoo I. Efficient framework for detection of
  version number attack in internet of things. In: Springer. ; 2018\string:
  480--492.

\bibitem{mirsky2018kitsune}
Mirsky Y, Doitshman T, Elovici Y, Shabtai A. Kitsune: an ensemble of
  autoencoders for online network intrusion detection. {\it arXiv preprint
  arXiv:1802.09089} 2018.

\bibitem{homoliak2019insight}
Homoliak I, Toffalini F, Guarnizo J, Elovici Y, Ochoa M. Insight into insiders
  and it: A survey of insider threat taxonomies, analysis, modeling, and
  countermeasures. {\it ACM Computing Surveys (CSUR)} 2019\string;
  52(2)\string: 1--40.

\bibitem{ahmad2021mitigating}
Ahmad MS, Shah SM. Mitigating Malicious Insider Attacks in the Internet of
  Things using Supervised Machine Learning Techniques. {\it Scalable Computing:
  Practice and Experience} 2021\string; 22(1)\string: 13--28.

\bibitem{khraisat2019novel}
Khraisat A, Gondal I, Vamplew P, Kamruzzaman J, Alazab A. A Novel Ensemble of
  Hybrid Intrusion Detection System for Detecting Internet of Things Attacks.
  {\it Electronics} 2019\string; 8(11)\string: 1210.

\bibitem{venkatraman2020adaptive}
Venkatraman S, Surendiran B. Adaptive hybrid intrusion detection system for
  crowd sourced multimedia internet of things systems. {\it Multimedia Tools
  and Applications} 2020\string; 79(5)\string: 3993--4010.

\bibitem{bovenzi2020hierarchical}
Bovenzi G, Aceto G, Ciuonzo D, Persico V, Pescap{\'e} A. A hierarchical hybrid
  intrusion detection approach in IoT scenarios. In: IEEE. ; 2020\string: 1--7.

\bibitem{sahay2020novel}
Sahay R, Geethakumari G, Mitra B. A novel blockchain based framework to secure
  IoT-LLNs against routing attacks. {\it Computing} 2020\string;
  102(11)\string: 2445--2470.

\bibitem{veeramakali2021intelligent}
Veeramakali T, Siva R, Sivakumar B, Senthil~Mahesh P, Krishnaraj N. An
  intelligent internet of things-based secure healthcare framework using
  blockchain technology with an optimal deep learning model. {\it The Journal
  of Supercomputing} 2021\string; 77(9)\string: 9576--9596.

\bibitem{babu2020sh}
Babu MJ, Reddy AR. SH-IDS: Specification Heuristics Based Intrusion Detection
  System for IoT Networks. {\it Wireless Personal Communications} 2020\string:
  1--23.

\bibitem{makkar2020efficient}
Makkar A, Kumar N. An efficient deep learning-based scheme for web spam
  detection in IoT environment. {\it Future Generation Computer Systems}
  2020\string: 467--487.

\bibitem{eskandari2020passban}
Eskandari M, Janjua ZH, Vecchio M, Antonelli F. Passban IDS: an intelligent
  anomaly-based intrusion detection system for IoT edge devices. {\it IEEE
  Internet of Things Journal} 2020\string; 7(8)\string: 6882--6897.

\bibitem{verma2020machine}
Verma A, Ranga V. Machine learning based intrusion detection systems for IoT
  applications. {\it Wireless Personal Communications} 2020\string;
  111(4)\string: 2287--2310.

\bibitem{sahu2021internet}
Sahu AK, Sharma S, Tanveer M, Raja R. Internet of Things attack detection using
  hybrid Deep Learning Model. {\it Computer Communications} 2021.

\bibitem{fotohi2021lightweight}
Fotohi R, Pakdel H. A Lightweight and Scalable Physical Layer Attack Detection
  Mechanism for the Internet of Things (IoT) Using Hybrid Security Schema. {\it
  Wireless Personal Communications} 2021\string: 1--18.

\bibitem{ambusaidi2016building}
Ambusaidi MA, He X, Nanda P, Tan Z. Building an intrusion detection system
  using a filter-based feature selection algorithm. {\it IEEE transactions on
  computers} 2016\string; 65(10)\string: 2986--2998.

\bibitem{shone2018deep}
Shone N, Ngoc TN, Phai VD, Shi Q. A deep learning approach to network intrusion
  detection. {\it IEEE transactions on emerging topics in computational
  intelligence} 2018\string; 2(1)\string: 41--50.

\bibitem{meidan2018n}
Meidan Y, Bohadana M, Mathov Y, et al. N-baiot—network-based detection of iot
  botnet attacks using deep autoencoders. {\it IEEE Pervasive Computing}
  2018\string; 17(3)\string: 12--22.

\bibitem{jan2019toward}
Jan SU, Ahmed S, Shakhov V, Koo I. Toward a lightweight intrusion detection
  system for the internet of things. {\it IEEE Access} 2019\string; 7\string:
  42450--42471.

\bibitem{zhou2020building}
Zhou Y, Cheng G, Jiang S, Dai M. Building an efficient intrusion detection
  system based on feature selection and ensemble classifier. {\it Computer
  networks} 2020\string; 174\string: 107247.

\bibitem{srinivas2022clustering}
Srinivas M, Patnaik MR. Clustering with a high-performance secure routing
  protocol for mobile ad hoc networks. {\it The Journal of Supercomputing}
  2022\string: 1--22.

\bibitem{sahay2022holistic}
Sahay R, Geethakumari G, Mitra B. A holistic framework for prediction of
  routing attacks in IoT-LLNs. {\it The Journal of Supercomputing} 2022\string;
  78(1)\string: 1409--1433.

\bibitem{benesty2009pearson}
Benesty J, Chen J, Huang Y, Cohen I. Pearson correlation coefficient. In:
  Springer.  2009 (pp. 1--4).

\bibitem{tibshirani1996regression}
Tibshirani R. Regression shrinkage and selection via the lasso. {\it Journal of
  the Royal Statistical Society: Series B (Methodological)} 1996\string;
  58(1)\string: 267--288.

\bibitem{dietterich2002ensemble}
Dietterich TG, others . Ensemble learning. {\it The handbook of brain theory
  and neural networks} 2002\string; 2(1)\string: 110--125.

\bibitem{bottou1995convergence}
Bottou L, Bengio Y. Convergence properties of the k-means algorithms. In: ;
  1995\string: 585--592.

\bibitem{bezdek1984fcm}
Bezdek JC, Ehrlich R, Full W. FCM: The fuzzy c-means clustering algorithm. {\it
  Computers \& geosciences} 1984\string; 10(2-3)\string: 191--203.

\bibitem{ankerst1999optics}
Ankerst M, Breunig MM, Kriegel HP, Sander J. OPTICS: Ordering points to
  identify the clustering structure. {\it ACM Sigmod record} 1999\string;
  28(2)\string: 49--60.

\bibitem{kim2003constructing}
Kim HC, Pang S, Je HM, Kim D, Bang SY. Constructing support vector machine
  ensemble. {\it Pattern recognition} 2003\string; 36(12)\string: 2757--2767.

\bibitem{kim2016long}
Kim J, Kim J, Thu HLT, Kim H. Long short term memory recurrent neural network
  classifier for intrusion detection. In: IEEE. ; 2016\string: 1--5.

\bibitem{friedman2017elements}
Friedman JH. {\it The elements of statistical learning: Data mining, inference,
  and prediction}.
\newblock springer open .
\newblock 2017.

\bibitem{hinton2006fast}
Hinton GE, Osindero S, Teh YW. A fast learning algorithm for deep belief nets.
  {\it Neural computation} 2006\string; 18(7)\string: 1527--1554.

\bibitem{hu2017survey}
Hu P, Dhelim S, Ning H, Qiu T. Survey on fog computing: architecture, key
  technologies, applications and open issues. {\it Journal of network and
  computer applications} 2017\string; 98\string: 27--42.

\bibitem{dataset}
NSL-KDD dataset. {\it https://www.unb.ca/cic/datasets/nsl.html (accessed:
  02.08.2020)}.

\bibitem{TONdataset}
N. Moustafa, ToN\_IoT Datasets, 2019 (online). {\it
  http://dx.doi.org/10.21227/fesz-dm97 (accessed: 02.08.2020)}.

\bibitem{gunupudi2017clapp}
Gunupudi RK, Nimmala M, Gugulothu N, Gali SR. CLAPP: A self constructing
  feature clustering approach for anomaly detection. {\it Future Generation
  Computer Systems} 2017\string; 74\string: 417--429.

\bibitem{su2020bat}
Su T, Sun H, Zhu J, Wang S, Li Y. BAT: Deep learning methods on network
  intrusion detection using NSL-KDD dataset. {\it IEEE Access} 2020\string;
  8\string: 29575--29585.

\bibitem{de2020hybrid}
Souza dCA, Westphall CB, Machado RB, Sobral JBM, Santos~Vieira dG. Hybrid
  approach to intrusion detection in fog-based IoT environments. {\it Computer
  Networks} 2020\string; 180\string: 107417.

\bibitem{tavallaee2009detailed}
Tavallaee M, Bagheri E, Lu W, Ghorbani AA. A detailed analysis of the KDD CUP
  99 data set. In: IEEE. ; 2009\string: 1--6.

\end{thebibliography}

\end{document}